\documentclass[12pt,reqno]{amsart}
\usepackage{fullpage}
\usepackage{amsfonts}
\usepackage{amsmath}
\numberwithin{equation}{section}
\usepackage{amsthm}
\usepackage{mathtools}
\usepackage{scalerel}
\usepackage{amssymb}

\usepackage{times}
\usepackage{graphicx}
\usepackage{subcaption}
\usepackage[export]{adjustbox}
 \usepackage{hyperref}
\usepackage{epsfig}
\usepackage{xcolor}
\usepackage{multirow}
\usepackage{enumerate}
\usepackage{booktabs}

\newtheorem{theorem}{Theorem}
\numberwithin{theorem}{section}

\newtheorem{lemma}[theorem]{Lemma}

\newtheorem{corollary}[theorem]{Corollary}

{ \theoremstyle{remark}\newtheorem{remark}[theorem]{Remark} }

\def\R{\mathbb{R}}
\def\C{\mathbb{C}}

\newcommand{\tr}{\operatorname{Tr}}

\begin{document}

\title[Article]{Gradual Eigenvector Ergodization in Coupled Ginibre matrices}

\author{ Margherita Disertori$^{1}$	and  Yan V Fyodorov$^{2}$	}
\address{$^1$ Institute for Applied Mathematics $\&$ Hausdorff Center for Mathematics, University of Bonn,Germany. Email: mdiserto@uni-bonn.de}
\address{$^2$ Department of Mathematics, King's College London, WC2R 2LS, UK. Email:yan.fyodorov@kcl.ac.uk}


\maketitle

\begin{abstract}
Non-Hermitian random matrices provide a useful framework for understanding universal characteristics of dissipative quantum chaotic systems with
loss or gain. We consider a model of two such systems represented by two independent $N\times N$ complex
Ginibre matrices interacting via a deterministic matrix $c{\bf 1}_N$,
where
 $c$ is the complex coupling parameter whose magnitude $|c|$ controls the interaction strength.
We characterize quantitatively how the eigenvectors of the whole system, initially localized in one of the individual
subsystems for $|c|=0$,  eventually spread over the full system with growing interaction strength.
The resulting asymptotic formula describing such spread in the limit $N\to \infty$ is very explicit and provides a full picture of the
gradual ergodization of eigenvectors as a function of the coupling parameter  $|c|$ in the whole transition regime.
As a by-product of our method we also compute the mean eigenvalue density for our model at the origin of the spectral bulk $z=0$
in the fully ergodic regime, when the coupling is scaled with the matrix size as $c=\sqrt{N}\tilde{c}$.
We find that as $N\to \infty$ the limiting density at the origin vanishes beyond the critical value $|\tilde{c}|=1.$ This 
 is compatible with the expected  split of the density support in the complex plane into two disjoint domains.

\end{abstract}

\section{Introduction}
In this paper we study the structure of the (right) eigenvectors ${\bf v}$  of an ensemble of $2N\times 2N$ random matrices $X$ composed of
four $N\times N$ blocks as shown below:
\begin{equation}\label{ensemble_def}
X=X (c,\bar c) :=\left(\begin{array}{cc} G_{1} & c\,{\bf 1}_N\\ \overline{c}\,{\bf 1}_N & G_{2}
\end{array}\right), \quad c\in \mathbb{C},  \quad {\bf v}=\left(\begin{array}{c}{\bf v}_1\\{\bf v}_2\end{array}\right),
\end{equation}
where  ${\bf 1}_N$ is the identity matrix, and $G_{j}\in  \mathbf{GinUE}_{j}(N), j=1,2$ are
two independent copies of the complex Ginibre matrices $G$ of size $N\times N$, i.e. square matrices
$G\in \C^{N\times N}$ with independent
identically distributed matrix elements 
$G_{j,k}\sim\mathcal{N}_{\C}(0,1).$

Such random matrices $X$ are generically non-Hermitian and non-normal, but diagonalizable with probability one, and it is natural to represent
their right eigenvectors in the column form as in Eq.(\ref{ensemble_def}), with ${\bf v}_{1,2}\in \mathbb{C}^N$ being two $N-$component
complex vectors. The $2N$ component eigenvector ${\bf v}$ will be assumed to be
normalized to the unit length: ${\bf v}^*{\bf v}={\bf v}_1^*{\bf v}_1+{\bf v}_2^*{\bf v}_2=1$, where we denote
$\mathbf{v}_{\sigma}^*= (\overline{v_{\sigma,1}}\ldots, \overline{v_{\sigma,N}}), \sigma=1,2$ for the Hermitian conjugate,
with bar standing for the complex conjugation.

For the symmetric model \eqref{ensemble_def} it will be convenient to denote by $\rho:=|c|^2$ the coupling strength. As long as $\rho=0$, the spectrum of $X$ consists of a superposition
of two independent sets of $N$ eigenvalues for $G_1$ and $G_2$ respectively, with eigenvectors ${\bf v}$ being such that either ${\bf v}_2$
or ${\bf v}_1$ are identically zero.
For $|c|^2>0$  it is easy to show that for each eigenvector ${\bf v}$ the two (squared) norms 
\begin{equation}\label{eq:pdef}
p_1:={\bf v}_1^*{\bf v}_1,\qquad p_2:={\bf v}_2^*{\bf v}_2
\end{equation}
must be simultaneously nonvanishing, inducing nontrivial correlations between the eigenvalues of the two sets.
One may however naturally expect that typically the two norms will remain noticeably different for small enough values of the coupling,
reflecting the "nonergodic"
nature of the eigenvectors inhomogeneously spreading across the full basis of size $2N$.
Only for strong enough coupling the "ergodization" of eigenvectors is expected to be complete.
Although qualitatively these features should exist for any finite $N\ge 2$,  one is usually mostly interested
in the large size limit $N\to \infty$, as only in this limit one may expect a certain universality of the behaviour,
i.e. insensitivity to the special detail of the distribution of the entries of the underlying two random matrices,
and, possibly, to the particular form of the off-diagonal coupling in the model.
One may also naturally expect that when 
complete eigenvector ergodization occurs in the limit $N\to \infty,$ 
the correlation properties of the eigenvalues of $X$ will be indistinguishable from those of a single Ginibre matrix of size $2N$. 

Investigating this type of problem seems not only interesting and challenging mathematically, but is also closely related to questions
actively discussed in the physics literature. As is well-known, Hermitian random matrices represent paradigmatic models for studying
spectral and eigenfunction characteristics of highly excited states of isolated quantum chaotic systems, both in single-particle\cite{Haake4th}
and more recently in many-body \cite{Prozen_etal_2018} setting. As was conjectured already long ago \cite{Deutsch,Srednicki} eigenstates
in quantum chaotic systems are expected to satisfy  a fundamental property of the
Eigenstate Thermalization Hypothesis (ETH), asserting that the highly excited eigenfunctions of a sufficiently disordered or
quantum chaotic system are uniformly distributed
in the phase space.
The studies of ETH-related questions keep attracting considerable attention, with many new interesting insights obtained
in recent years, see e.g. \cite{Foini_Kurchan19,Pappalardi_etal2022}. In mathematics, ETH translates into ergodicity of the
wavefunctions in single-particle quantum chaotic systems and is the subject of a long history of studies
starting from  Shnirelman’s theorem \cite{Shnirel74}, see e.g. \cite{Rudnick94,Zelditch87}.
A closely related mathematical paradigm in single-particle quantum chaos is quantum ergodicity and its stronger form, Quantum Unique Ergodicity.
Due to the relation between quantum chaos and random matrix theory (RMT), there was a considerable interest in studying
ergodicity-related  questions in the context of RMT, see e.g.
\cite{Cipol21,Benigni22,ErdosRiabov24}. In that framework, it was suggested long ago by Grobe, Haake and Sommers (GHS) \cite{grobe1988quantum}
that non-Hermitian matrices from the complex Ginibre ensemble are expected to represent a faithful model for quantum chaotic systems where
mechanisms of gain and/or loss due to dissipation are operative, see \cite{Haake4th} for a more detailed discussion.
Despite the fact that precise conditions of applicability and universality of the GHS conjecture remain 
a matter of controversy and active discussion \cite{villasenor2024breakdown,villasenor2025correspondence,mondal2026transient},
their work attracted renewed attention to investigate spectral characteristics of non-Hermitian matrices.
Interest in non-Hermitian quantum chaos, especially in a many-body setting,  and in the associated random matrix description 
intensified in recent years, see e.g. \cite{ProsenDissip,Prasad_etal_2022,Garcia22,Garcia23,Cipolloni-KF1,Cipolloni-KF2}.
In this context, the model \eqref{ensemble_def} of two coupled Ginibre matrices is expected to describe a pair of interacting non-Hermitian/dissipative quantum chaotic
systems.

In the present paper we consider the  joint probability density (JPD) of the variables
$p_1,p_2\in [0,1]^{2}$ with respect to the measure $dp_{1}dp_{2}$ at the spectral point $z\in \C,$ defined by
\begin{equation}\label{def:JPD-finiteN}
\pi_{N}(z,p_1,p_2)= \pi_{N}^{c,\bar c}(z,p_1,p_2):=\frac{\mathbb{E}[\Pi_N(z,p_1,p_2)]}{\mathbb{E}[\rho_N(z)]}=
\frac{\mathbb{E}[\Pi_N(z,p_1,p_2)]}{\int_{[0,1]^{2}}\mathbb{E}[\Pi_N(z,p_1,p_{2})]  dp_{1}dp_{2}},
\end{equation}
where $\mathbb{E}$ denotes the average  over the ensembles of ${\small \mathbf{GinUE}_1(N)}$ and  ${\small \mathbf{GinUE}_2(N)},$
\begin{equation}\label{def:rho}
\rho_N(z):=  \sum_{a=1}^{2N}\delta^{(2)}(z-z_a)
\end{equation}
is the unnormalized empirical density of the complex eigenvalues $z_a, \, a=1,\ldots, 2N$ for the matrix $X,$ and
\begin{equation}\label{def:Pi}
\Pi_N(z,p_1,p_2)=:\sum_{a=1}^{2N}\delta^{(2)}(z-z_a)\,
\delta\big (p_1- {\bf v}_{1} (a)^* {\bf v}_{ 1} (a)\big )\, \delta\big(p_2-{\bf v}_{2} (a)^*{\bf v}_{2} (a)\big),
\end{equation}
is the unnormalized joint empirical density of the squared norms $p_1$ and $p_2$  introduced in \eqref{eq:pdef} corresponding to those
eigenvectors ${\bf v} (a)$ of $X$
whose eigenvalues $z_a$ are in the vicinity of a point $z\in \mathbb{C}.$ Here we used the convention $\delta^{(2)}(z):=\delta(\Im z)\delta(\Re z)$,
with $\delta(x)$ being the standard Dirac delta-function and we assumed
$dzd\overline{z}:=d(\Im z)d(\Re z)$.

Using the recently proposed approach \cite{YF2025} to the joint empirical density of an eigenvalue $z$ and the corresponding right eigenvector
$\mathbf{v}$ we derive a representation for the  normalized joint probability density at spectral point $z=0,$  and finite $N$
(see  Corollary \ref{Cor2.2} below) which is well suited to study the
$N\to \infty$ asymptotics. Our main result is an explicit formula for the limiting JPD at the spectral point $z=0.$
\begin{theorem}\label{main_asy}
For any fixed, $N-$independent, value of the coupling $0<|c|<\infty,$ the normalized limiting JPD at the spectral point $z=0$ is given by
\begin{align}\label{main_asy_explicit}
\pi_{\infty}(z=0,p_1,p_2)=&\pi_{\infty}^{c,\bar c}(z=0,p_1,p_2): =  \lim_{N\to \infty}\pi_{N}^{c,\bar c}(0,p_1,p_2)\\
&=
\frac{|c|^2}{p_1p_2}\,e^{-|c|^2\left(\frac{p_1}{p_2}+\frac{p_2}{p_1}\right)}\delta(p_1+p_2-1)\left[I_1(2|c|^2)+\frac{1}{2}\left(\frac{p_1}{p_2}+\frac{p_2}{p_1}\right)
I_0(2|c|^2)\right],\nonumber
\end{align}
where $I_n(x)$ is the modified Bessel function of the first kind 
\begin{equation} \label{Bessel_I}
 I_n(x)=\int_0^{2\pi}e^{x\cos{\theta}}\cos{n\theta}\,\frac{d\theta}{2\pi},\quad n\in \mathbb{N}_{0}.
 \end{equation}
In addition, the asymptotic mean eigenvalue density at $z=0$ is given by 
\begin{equation}\label{main-asy-ED}
\mathcal{D}_{\infty} (0):=\lim_{N\to \infty}\mathbb{E}[\rho_N(0)]= \frac{2}{\pi }.
\end{equation}
\end{theorem}
\begin{remark}
The parameter $\bar c$ in the random matrix  $X (c,\bar c)$ from \eqref{ensemble_def} was chosen by convenience and
can be replaced by any complex number $d\in \mathbb{C}\setminus\{0 \}.$ The new random matrix takes the form
\begin{equation}\label{ensemble_def1}
X (c,d) :=\left(\begin{array}{cc} G_{1} & c\,{\bf 1}_N\\ d\,{\bf 1}_N & G_{2}
\end{array}\right), \qquad c,d\in \mathbb{C}\setminus\{0 \}.
\end{equation}
The asymptotic formula  \eqref{main_asy_explicit} extends directly to this case by a similarity transformation.
This is the content of the following main corollary.
\end{remark}

\begin{corollary}\label{cor:maincor} 
For any fixed, $N-$independent, coupling $ c,d\in \mathbb{C}\setminus\{0 \},$ the normalized limiting JPD of $X (c,d)$ at the spectral point
$z=0$ is given by
\begin{align}\label{main_asy_explicit1}
\pi_{\infty}^{c,d}(z=0,p_1,p_2)&: =  \lim_{N\to \infty}\pi_{N}^{c,d}(0,p_1,p_2)\\
&=
\frac{\rho}{p_1p_2}\,e^{-\rho \left(\frac{1}{\gamma }\frac{p_1}{p_2}+\gamma\frac{p_2}{p_1}\right)}\delta(p_1+p_2-1)
\left[I_1(2\rho )+\frac{1}{2}\left(\frac{1}{\gamma}\frac{p_1}{p_2}+\gamma \frac{p_2}{p_1}\right)
I_0(2\rho)\right],\nonumber
\end{align}
where $\gamma:=\frac{|c|}{|d|}$  is the  directional coupling,  $\rho:=|cd|$
is the symmetric coupling and 
 $I_n(x)$ is the modified Bessel function of the first kind \eqref{Bessel_I}.
In addition, the asymptotic mean eigenvalue density at $z=0$ is given by
\begin{equation}\label{main-asy-ED1}
\mathcal{D}_{\infty}^{c,d} (0):=\lim_{N\to \infty}\mathbb{E}[\rho_N^{c,d}(0)]= \frac{2}{\pi }.
\end{equation}
\end{corollary}
 \begin{proof}[Proof of Corollary \ref{cor:maincor}]
Note that, by construction,  $\pi_{N}^{c,d} (z,p_{1},p_{2})= f (z,p_{1},p_{2}) \delta (p_{1}+p_{2}-1)$ for some function $f.$
Let $\alpha\in \mathbb{C}$ be any complex fourth root of $d/c$
and set
\[
S_{\alpha }:= \begin{pmatrix}
\alpha {\bf 1}_N& 0\\
0 & \alpha^{-1}{\bf 1}_N
\end{pmatrix}.
\]
Define in addition $g:= \alpha^{2}c= \alpha^{-2}d= |g|\, e^{i\tau }$, with $|g|>0, \tau \in (-\pi,\pi ].$
By  direct computation we obtain   $g^{2}=cd$ and
\[
X (c,d)= \begin{pmatrix}G_{1} & c\,{\bf 1}_N
\\ d\,{\bf 1}_N & G_{2}
\end{pmatrix}=  S_{\alpha }^{-1} \begin{pmatrix}G_{1} & g\,{\bf 1}_N
\\ g\,{\bf 1}_N & G_{2}
\end{pmatrix} S_{\alpha}=
e^{i\tau } S_{\alpha }^{-1} \begin{pmatrix}G_{1}e^{-i\tau } & |g|\,{\bf 1}_N
\\ |g|\,{\bf 1}_N & G_{2}e^{-i\tau }
\end{pmatrix} S_{\alpha}.
\]
The complex Ginibre law is invariant under multiplication by a phase so 
\[
 \begin{pmatrix}G_{1}e^{-i\tau } & |g|\,{\bf 1}_N
\\ |g|\,{\bf 1}_N & G_{2}e^{-i\tau }
\end{pmatrix}\ =:\tilde{X} \overset{d}{=} X (|g|,|g|),
\]
which corresponds to the case \eqref{ensemble_def} with $c=\bar c=|g|.$ Moreover 
$\mathbf{v}$ is a normalized right eigenvector of $X (c,d)$ at the spectral point $z$ iff
$\| S_{\alpha }\mathbf{v}\|^{-1} (S_{\alpha }\mathbf{v})$ is a normalized right eigenvector of $\tilde{X}$ at the spectral point $e^{-i\tau }z.$
Using \eqref{def:JPD-finiteN} and $|\alpha^{2}|=|g|/|c|, |\alpha^{-2}|=|g|/|d|,$ we obtain 
\[
\pi_{N}^{c,d} (z,p_{1},p_{2}) dp_{1}dp_{2}= \pi^{|g|,|g|}_{N} (e^{-i\tau }z,q_{1},q_{2}) dq_{1}dq_{2}=  \pi^{|g|,|g|}_{N} (e^{-i\tau }z,q_{1},q_{2})
\frac{q_{1}q_{2}}{p_{1}p_{2}} dp_{1}dp_{2}
\]
where
\[
q_{1}:=\frac{|d|p_{1}}{|d|p_{1}+|c|p_{2}},\quad  q_{2}:=\frac{|c|p_{2}}{|d|p_{1}+|c|p_{2}},
\]
and we used 
\[
\frac{dp_{1}dp_{2}}{p_{1}p_{2}}\delta (p_{1}+p_{2}-1) =\frac{dq_{1}dq_{2}}{q_{1}q_{2}} \delta (q_{1}+q_{2}-1).
\]
Formula \eqref{main_asy_explicit1} now follows replacing \eqref{main_asy_explicit} in $ \pi^{|g|,|g|}_{\infty} (0,q_{1},q_{2})\frac{q_{1}q_{2}}{p_{1}p_{2}}.$ Formula \eqref{main-asy-ED1} follows combining 
  $\int \pi_{N}^{c,d} (0,p_{1},p_{2}) dp_{1}dp_{2}=\int \pi_{N}^{|g|,|g|} (0,q_{1},q_{2}) dq_{1}dq_{2}$
  with \eqref{main-asy-ED}.
\end{proof}

\begin{remark}[Effect of the symmetric and directional coupling $\rho,\gamma$]
The asymmetric occupation bias  $\gamma $
 disappears precisely when the two directional couplings have equal
magnitude  $|d|=|c|,$ which does \textit{not} imply $d=c.$  For $|d|=|c|$ we obtain 
\begin{equation}\label{eq:pidp}
\pi_{\infty}^{c,d}(z=0,p_1,p_2)=\pi_{\infty}^{|c|,|c|}(z=0,p_1,p_2)=\pi_{\infty}^{c,\bar c}(z=0,p_1,p_2).
\end{equation}
This JPD is symmetric under $p_{1}\leftrightarrow p_{2}$ and depends only on the 
symmetric coupling  $\rho=|cd|$.
The loss of subsystem identity as a function of $\rho$ and $\gamma$ is most directly visible in the density $P_{\rho,\gamma } (p)$
with respect to $dp=dp_{1},$ which is obtained integrating over $p_{2}:$ 
\begin{equation}\label{eq:pidy}
P_{\rho,\gamma } (p)= \frac{\rho}{p (1-p)}\,e^{-\rho \left(\frac{1}{\gamma }\frac{p}{1-p}+\gamma\frac{1-p}{p}\right)} \left[I_1(2\rho )+\frac{1}{2}\left(\frac{1}{\gamma}\frac{p}{1-p}+\gamma \frac{1-p}{p}\right) I_0(2\rho)\right],
\end{equation}
This function is plotted in Figure~\eqref{fig1} for several values f $\rho $ and $\gamma$.

In the symmetric case $\gamma=1$ the formula simplies to
\[
P_{\rho,1} (p)=\rho \Big (S (p)+2 \Big)\, \,e^{-\rho S (p )} \left[I_1(2\rho )+\frac{1}{2}S (p ) I_0(2\rho)\right]
\]
where $S(p):=\frac{p}{1-p}+\frac{1-p}{p}\geq 2$
and we used the identity $[p(1-p)]^{-1}=S (p) +2.$
We compute
\[
\partial_p\log P_{\rho,1}(p)
=F(\rho,S(p))\,\frac{2p-1}{p^2(1-p)^2},
\]
where
\[
F(\rho,S):=
\frac1{S+2}
+\frac{\frac12 I_0(2\rho)}
{I_1(2\rho)+\frac S2I_0(2\rho)}-\rho.
\]
For fixed $\rho>0$, the function $S\mapsto F(\rho,S)$ is strictly decreasing
and tends to $-\rho$ as $S\to\infty$.  If $F(\rho,2)>0$, the density has two
symmetric maxima separated by a minimum at $p=1/2$, corresponding to
eigenvectors still preferentially occupying one of the two subsystems.
If $F(\rho,2)\leq0$, it has a single maximum at $p=1/2$, signalling the loss
of subsystem identity.
The function 
\[
\rho \mapsto F(\rho,2)=\frac14+
\frac{I_0(2\rho)}{2[I_0(2\rho)+I_1(2\rho)]}-\rho
\]
is strictly decreasing, since $ I_1(x)/I_0(x)$ is strictly
increasing for $x>0.$ Moreover, it  satisfies $F (0,2)=3/4$ and  converges to $-\infty$ as $\rho\to \infty$.
Hence the equation $F(\rho,2)=0$ has a unique positive solution, which we call the
\emph{subsystem identity threshold}. Equivalently, this threshold is determined by
\begin{equation}\label{eq:identity-threshold}
4\rho[I_0(2\rho)+I_1(2\rho)]
=3I_0(2\rho)+I_1(2\rho).
\end{equation}
Numerical solution of this equation yields
\[
\rho_{id}(1)=0.5828789968\ldots\, .
\]
For $\gamma\neq1$, two maxima still appear at small $\rho,$ but they no longer need to be symmetric.
As $\rho$ increases, one maximum merges with a local minimum and disappears at a point
$(\rho_{id} (\gamma), p_* (\gamma))$ determined by
\[
\partial_pP_{\rho_{id} ,\gamma}(p_*)=0=\partial_p^2P_{\rho_{id},\gamma}(p_*).
\]
These equations can be solved numerically.  For the example $\gamma=2$ shown
in Fig.~\eqref{fig1} , we obtain
\[
\rho_{id}(2)=0.3357256676\ldots,
\qquad
p_* (2) =0.4469829585\ldots\, .
\]
The remaining maximum subsequently approaches the strong-coupling balance
point $p_\gamma=\gamma/(1+\gamma)$. In particular, for $\gamma=2$ we have
$p_\gamma=2/3.$ Thus $p_\gamma$ describes the location of the peak in the
large-$\rho$ limit, not the point where the two stationary points meet.

\noindent One may hope that such a loss-of-identity ergodization transition may be observed
in numerical simulations of models of many-body chaotic non-Hermitian systems, e.g. a
pair of coupled
non-Hermitian Sachdev--Ye--Kitaev models, see \cite{Garcia22,Garcia23}.
\end{remark}

\begin{remark}[Limit behavior at small and large coupling]
In the limit of vanishing coupling strength $|c| \to 0$ and $|d|\to 0$ we have, 
 using  $I_0(2\rho )\approx 1$ and $I_1(2\rho )\approx \rho,$
\begin{align*}
\lim_{|c|,|d| \to 0} P_{\rho,\gamma } (p)&=
\lim_{|c|,|d| \to 0} \frac{1}{p (1-p)}\,e^{-\left(|d|^{2}\frac{p}{1-p}+|c|^{2}\frac{1-p}{p}\right)} \frac{1}{2}\left(|d|^{2}\frac{p}{1-p}+|c|^{2} \frac{1-p}{p}\right)\\
&=
\frac{1}{2}\left[\delta(1-p)+\delta(p)\right]
\end{align*}
in distribution. This describes the complete localization of eigenvectors in either of the two subsystems (cf. Figure~\ref{fig1} for small values of $\rho $). The limit does not depend on the relative convergence speed of $c,d.$ A singular one-way limit arises when only one directional coupling vanishes.
In this last case, the limit retains a nontrivial mixed component.
\medskip

\noindent In the limit of large symmetric coupling  $\rho \gg 1,$ with $\gamma$ fixed,
one may use the asymptotics $I_0(2\rho )=I_1(2\rho )\approx \frac{e^{2\rho }}{2\sqrt{\rho \pi}}$ and Laplace method to obtain
\begin{align*}
P_{\rho,\gamma } (p)&\approx\frac{\sqrt{\rho }}{2\sqrt{\pi }\,p (1-p)}\,
e^{-\rho \left[\frac{1}{\gamma }\frac{p}{1-p}+\gamma\frac{1-p}{p}-2\right]} \left[1+\frac{1}{2}\left(\frac{1}{\gamma}\frac{p}{1-p}+\gamma \frac{1-p}{p}\right)\right]\\
&
\approx \sqrt{\frac{\rho (1+\gamma)^{4}}{\pi \gamma^{2}}} e^{-\frac{\rho (1+\gamma)^{4}}{\gamma^{2}}
\left(p-p_{\gamma} \right)^{2}}
\end{align*}
so that the JPD becomes essentially a Gaussian distribution with variance proportional to
$\rho^{-1}\ll 1$, sharply peaked around value  $p=p_{\gamma }:=\frac{\gamma }{1+\gamma }$ (cf. Figure~\ref{fig1} for large values of $\rho $).
This result signals complete ergodization between the two subsystems.
Note that the concentration point $p_{\gamma}$ equals $\frac{1}{2}$ only when the
directional bias disappears  $\gamma =1.$

\begin{figure}[t]
  \centering
  \includegraphics[width=0.98\textwidth]{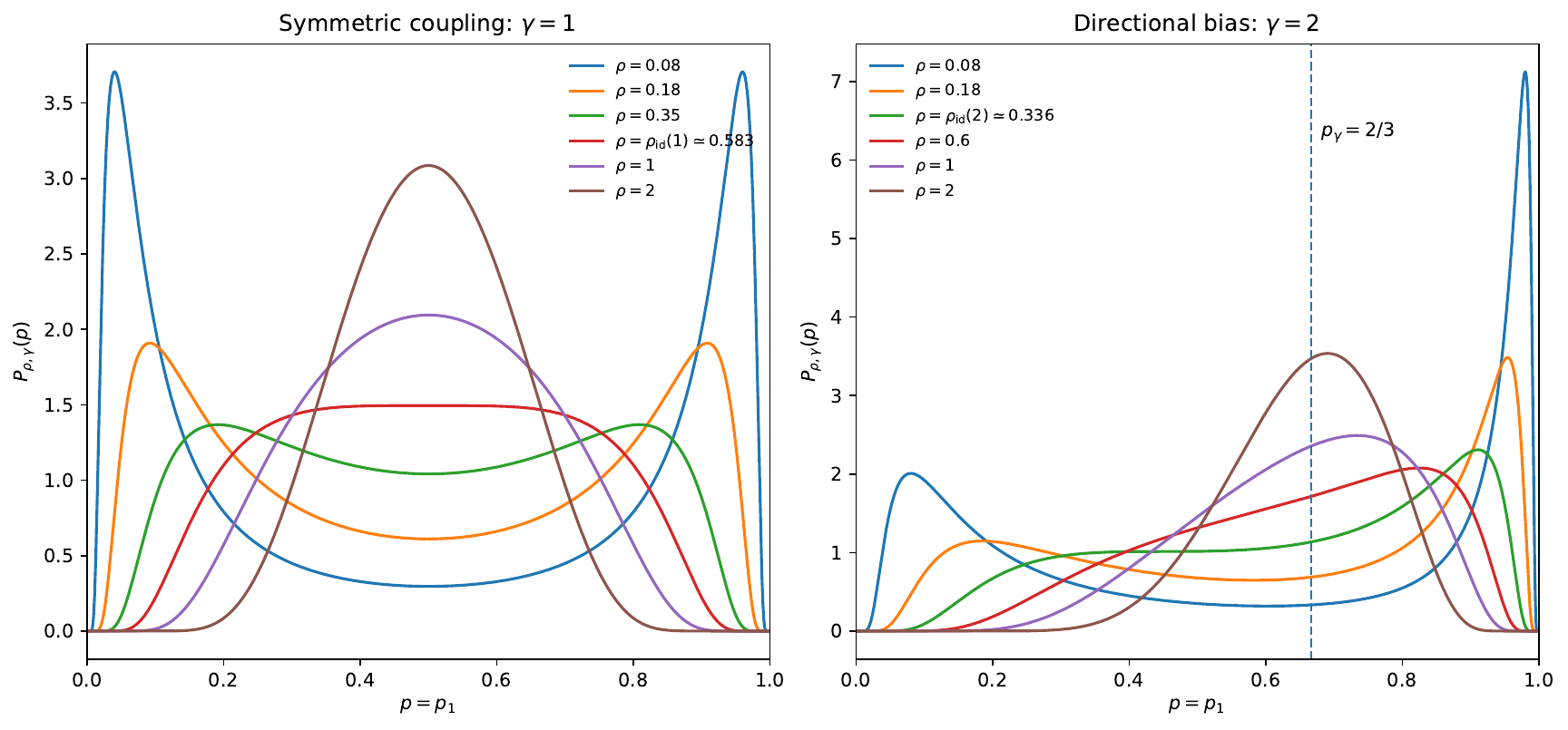}
  \caption{Limiting density $P_{\rho,\gamma}(p)$ from \eqref{eq:pidy}. Left: the symmetric case $\gamma=1$, including the identity threshold $\rho_{\mathrm{id}}(1)\simeq0.583$ from \eqref{eq:identity-threshold}. Right: the directionally biased case $\gamma=2$, including the directed threshold $\rho_{\mathrm{id}}(2)\simeq0.336$ at which the secondary peak disappears. The dashed line marks the large-coupling balance point $p_\gamma=\gamma/(1+\gamma)=2/3$. Increasing $\rho$ drives the occupation law from a two-peak profile towards a single hybridized peak.}
  \label{fig1}
\end{figure}

\end{remark}

\begin{remark}
The restriction of the above expressions to the spectral point at the origin $z=0$ is purely technical and is chosen for the sake of simplicity
and transparency of the derivation. Essentially the same formulas are expected to be valid for any point $z$ in the bulk of the unperturbed Ginibre
spectrum, i.e. as long as $\lim_{N\to \infty} |z|/\sqrt{N}<1$.
It would be also natural to conjecture the result remains valid when the  complex gaussian Ginibre ensemble $GinUE$ is replaced by a
more general non-Hermitian ensemble with independent, identically distributed entries and fourth moment matching conditions like in
\cite{TaoVu2008,TaoVu2015}, and/or  the nature of the coupling is changed.
Namely, denoting $\hat{C}$ one off-diagonal $N\times N$ block of the matrix in $X$, one may expect that the
result will remain true replacing the coupling $|c|$ with   $N^{-1/2}\|\hat{C}\|_{HS}$, where $\|\hat{C}\|_{HS}=\sqrt{\sum_{ij}|c_{ij}|^2}$ is the
Hilbert-Schmidt
norm of the coupling block. The background for such a conjecture should be clear from the discussion below.
\end{remark}

\begin{remark}
As a by-product of our derivation (see \eqref{main-asy-ED} \eqref{main-asy-ED1}) we find that everywhere in the transition region the
asymptotic mean eigenvalue density  is independent of the couplings $\rho,\gamma$ and equal to the twice mean limiting
density of a single Ginibre ensemble : $\mathcal{D}_{\infty}^{c,d}(0)=\frac{2}{\pi}$. Such density is expected to retain asymptotically
the same value everywhere in the disk $|z|<\sqrt{N}$, so that multiplying it with the area of the disk gives $2N$, the total number of eigenvalues. We thus see that a change in the mean eigenvalue density is only possible in the fully ergodic regime when
the coupling grows with size  $N$ as $\rho =N\tilde{\rho}$, with $\tilde{\rho}>0$ of the order of unity.
In the latter case our calculation at the origin identifies a transition at $\tilde{\rho}=1$ which is compatible with the expected splitting of the global mean-density support into two disjoint domains for $\tilde{\rho}>1$. We obtain (cf. Section~\ref{subs:asext})  
\[
\lim_{N\to \infty}\mathbb{E}[\rho_N^{c,d}(0)]=\lim_{N\to \infty}\mathbb{E}[\rho_N^{\sqrt{\tilde{\rho }N},\sqrt{\tilde{\rho }N}}(0)] =\frac{2}{\pi}(1-\tilde{\rho})\qquad \forall \tilde{\rho}<1,
\]
while the limit is zero for $\tilde{\rho}>1.$ This is compatible with the expected transition in the support.
\end{remark}

\paragraph{\textit{Discussion on the  results}} To put our results in the appropriate context of related studies, one should mention that the questions of eigenfunction
ergodization in coupled non-dissipative quantum chaotic systems attracted already some interest both in physics
and mathematics literature from various perspectives. In particular, an impressive recent paper \cite{Stone_etal_2025}
studied in detail the model of a finite number $D$ of coupled Hermitian random matrices arranged on a circle
(a natural generalization of the structure
displayed in Eq.(\ref{ensemble_def}), our case being $D=2$)
  and addressed the ergodization of eigenvectors with increased coupling between nearest neighbours.
  Using the random Hermitian
  Wigner matrices $V$ in place of our GinUE, 
   the authors introduced a control parameter $\Gamma$ which in our notations reads
   $\Gamma=\left(\frac{N^{1/2}\|\hat{C}\|_{HS}}{\|\hat{V}\|_{HS}}\right)^2$  and is the analog of the parameter $|c|^{2}.$
They then demonstrate that if $\Gamma$ grows as $N^{\epsilon}, \epsilon>0$ the eigenvectors are fully ergodized, so that in our $D=2$ case
$p_1=p_2=1/2(1-o(1))$.  At the same time, if $\Gamma$ behaves as $  N^{-\epsilon}, \epsilon>0$ then $p_1=1-o(1), p_2=o(1)$ or vice versa,
signalling the total localization of each eigenfunction inside one of the subsystems. In this picture a nontrivial partial ergodization
described by $p_1$ and $p_2$ of the same order but significantly statistically
different from $1/2$ may occur only if $\Gamma=O(1)$ as $N\to \infty$. Replacing in our case $V$ with Ginibre matrix whose Hilbert-Schmidt norm
is $\|\hat{G}\|_{HS}=\sqrt{\sum_{ij}1}=N$ while for the diagonal perturbation $ \|\hat{C}\|_{HS}=|c|\sqrt{N}$ we see that gradual ergodization
transition should happen when  $\Gamma=|c|^2=O(1)$ confirming the range chosen in our paper.

  Describing a situation within such a "transition to global chaos" regime characterized by an incomplete/partial spread
of the eigenvectors over the full basis seems however to go beyond the remit of the methods exploited in
\cite{Stone_etal_2025}. One may expect that describing that situation in detail should require exploiting a
non-perturbative approach, with the supersymmetric nonlinear $\sigma$-model description due to Efetov \cite{Efebook}
looking as a natural candidate.
In fact, exactly the same Hermitian, GUE-based model for $D=2$ case was treated by the standard supersymmetry
approach at physical level of rigour long ago \cite{Efe2000}. It yielded an explicit, albeit involved, expression in the
transition regime for the joint probability density of two variables closely related to $p_1$ and $p_2$; in our
terminology, the JPD of variables $\tilde{p}_{1k}=N|{\bf v}_{1,k} |^2$ and $\tilde{p}_{2l}=N|{\bf v}_{2,l}|^2$,
where ${\bf v}_{1,k}$ and ${\bf v}_{2,l}$ are two arbitrary components of the vectors ${\bf v}_1$ and ${\bf v}_2$,
respectively. Such JPD is  independent of the choice of indices $1\le (k,l)\le N,$ since inside each half of
the system there must be always full ergodization for any coupling.  Following the essential recent progress in rigorous
version of the Efetov supersymmetry approach, see \cite{MSchTSch23}, one may expect that  the corresponding derivation could be
in principle made mathematically rigorous, though at the present moment it has not been yet verified. 

Coming back to the non-Hermitian problems, though we are unaware of any previous studies in this direction, it was natural to expect a qualitatively
similar picture of ergodization to happen as well in that case. Our equations Eqs.~\eqref{main_asy_explicit}
\eqref{main_asy_explicit1} \eqref{eq:pidy}
confirm this intuition and, moreover, provide a fully quantitative description of the eigenvector ergodization transition
controlled by the parameters $\rho,\gamma$. A technical advantage of using complex Ginibre rather than Hermitian GUE matrices for diagonal blocks is that in such a case one may readily employ
a mathematically rigorous approach proposed very recently by the second author in \cite{YF2025}.  Such an approach was argued to provide a certain
useful technical alternative to both the so-called Girko approach to non-Hermitian random matrices \cite{Girko1984}  and to the Efetov-type
supersymmetry approach, effectively combining some features of
the two. The essence of the approach suggested in \cite{YF2025} is briefly summarized in the next section, with details of its implementation for
the  present problem given in further sections of the paper and Appendices. We find it highly rewarding that the resulting formulas are very explicit
and we believe elegant, and provide a transparent description of the transition regime and the "loss of identity" phenomenon.
One may also hope that a generalization of the method may eventually be used to describe the highly non-trivial eigenvalue correlations arising in
the transition regime. The corresponding treatment seems technically much more challenging but not impossible. Finally, let us mention that a problem
of eigenfunction ergodization
is intimately related to a more challenging problem of Anderson (de)localization of eigenfunctions. Recently there was an essential progress
\cite{MSchTSch23,YauYin2025, Dubova_etal_2025a,Dubova_etal_2025b,ErdosRiabov25,Drogin2025}  towards
rigorously describing delocalization in the framework of a special class of random matrices - random banded matrices - introduced for that purpose in
\cite{bandCasati etal} and studied at a physical level of rigour using Efetov's nonlinear $\sigma-$ model approach long ago \cite{FM1991,FM1994}.
Non-Hermitian analogues of random banded matrices have been addressed mainly numerically\cite{GhoshKulkarniRoy2023}, while mathematically rigorous
analysis has only started recently\cite{MSchTSch25}. We hope that the methods of the present paper may help to further boost progress in this
subject.\vspace{0,2cm}

\paragraph{\textbf{Organization of this paper.}} In  Section \ref{sect2} we derive an integral representation for mean  joint empirical density 
$\mathbb{E}[\Pi_N(z,p_1,p_2)]$ at finite $N$ (see  Theorem \ref{Thm2.1} and Corollary \ref{Cor2.2}). 
The outline of the strategy for a general random matrix $X$ is given in Section \ref{sect:2.1}.
The implementation of the strategy in the case of the two coupled Ginibre random matrices  Eq.(\ref{ensemble_def}),
is worked out in Section \ref{sect:2.2}.
Extracting the leading asymptotics as $N\to \infty$ for a fixed coupling $|c|^2$ and proving Theorem\ref{main_asy}
is relegated to Section \ref{sec:3.1}. Finally, the computation in the regime of $|c|\sim \sqrt{N}$ is presented in Section~\ref{subs:asext}.
Some technical results, needed in the proof of Theorem \ref{main_asy}, are presented in the Appendix.
\vspace{0,2cm}

\paragraph{\textbf{ Acknowledgements}}
The research at the University of Bonn was supported by the German Research Foundation (DFG) under Germany's Excellence Strategy-GZ 2047/1 - 390685813, and by the  Deutsche Forschungsgemeinschaft (DFG, German Research Foundation) – CRC 1720 – 539309657.
The research at King's College London was supported by  EPSRC grant {\bf UKRI1015} "Non-Hermitian random matrices: theory and applications". YVF is most grateful to Noemi Cuppone for her assistance with using Mathematica software which was instrumental at intermediate stages of the project and to Dmitry Savin, Mira Shamis and Mariya Shcherbina for discussions and interest in the paper. YVF acknowledges with thanks the kind hospitality extended to him by the
Institute for Applied Mathematics at the University of Bonn during inception and completion of the project,
and is especially grateful to Prof. Dr. Anton Bovier
and Prof. Dr. Patrik Ferrari for allocating funds which made his research stay possible.

\section{Integral representation for the averaged JPD at finite $N$}
\label{sect2}

\subsection{Strategy for a general random matrix}
\label{sect:2.1}
Consider a complex $n\times n$   random matrix  $X\in \C^{n\times n}$ with probability measure $\mathbb{P}.$
We denote by $\left \langle \cdot \right\rangle_X$ the corresponding average and assume all moments of the measure
are finite. This is for example true if the random elements of $X$ are independent and normally distributed,
as in the case considered in this paper.

We will use the following formula for the joint empirical density of an eigenvalue $z$ and the corresponding right eigenvector ${\bf v},$
which was proved in  \cite{YF2025}.
\begin{theorem} \cite[Theorem 1.1]{YF2025} \label{MainMetaTheorem}
Consider a square matrix  $X\in \C^{n\times n}$
with $n$ distinct eigenvalues $z_a,\,  a=1, \ldots, n$ and associated right eigenvectors
${\bf v} (a)\in \C^{n}.$ Each eigenvector ${\bf v} (a)$ is  normalized  by  ${\bf v} (a)^{*}{\bf v} (a) =1$ and its
first entry ${\bf v} (a)_{1}$ has  the fixed ($a-$ independent) phase $\theta_a=\theta,$ with $\theta$  uniformly distributed on
$[0,2\pi).$
Let  $d\mu_H({\bf v})$ be the associated invariant Haar measure on the unit complex sphere $\mathbf{S}^c_{n}$,
which can be explicitly defined via
$d\mu_H({\bf v})=C_{n}\delta({\bf v}^*{\bf v}-1)d{\bf v}d{\bf v}^*$, where
$C_{n}=\frac{(n-1)!}{\pi^{n}}>0$ is the normalization constant 
$\int_{\mathbf{S}^c_{n}}d\mu_H({\bf v})=1$ and
$d{\bf v}d{\bf v}^*:=\prod_{j=1}^{n}d{v}_jd\overline{v}_j$.

Then the (unnormalized) joint empirical density of an eigenvalue $z$ and the corresponding right eigenvector ${\bf v}$  with respect to the measure
 $dzd\overline{z}\,d\mu_H({\bf v})/C_n$, further
  averaged over the phase $\theta$ uniformly distributed in $[0,2\pi)$
  is given by
\begin{equation}\label{Meta_Kac_Rice_right_eigenvec}
\mathcal{P}_{n}(z,{\bf v})=
\frac{1}{\pi}\delta^{(2n)}\left(\left(X-z{\bf 1}\right){\bf v}\right)\,\left|\frac{d}{dz}\det\left(X-z{\bf 1}\right)\right|^2.
\end{equation}
\end{theorem}
\noindent The expression above  can be reformulated as  
\[
\mathcal{P}_{n}(z,{\bf v})= \frac{(-1)^{n}}{\pi}\lim_{|w-z|\to 0}\frac{\partial^2}{\partial w\partial \overline{w}} {\small
\left[\delta^{(2n)}\left(\left(z{\bf 1}-X\right){\bf v}\right)\,
\det{ Y_X(w,\overline{w})}\right]},
\]
where  $Y_X(z,\overline{z})$ is the so called  \emph{hermitized matrix} defined as
\begin{equation}\label{Hermitized}
Y_X(z,\overline{z})=\left(\begin{array}{cc} 0 &  z{\bf 1}-X\\  \overline{z}{\bf 1}-X^{*} & 0\end{array}\right).
\end{equation}
Averaging over the randomness and using the boundedness of the moments we obtain
\begin{equation}\label{main_practical}
\left\langle\mathcal{P}_n(z,{\bf v})\right\rangle_{X}= \frac{(-1)^{n}}{\pi}\lim_{|w-z|\to 0}\frac{\partial^2}{\partial w\partial \bar w}\left\langle {\small
\left[\delta^{(2n)}\left(\left(z{\bf 1}-X\right){\bf v}\right)\,
\det{Y_X(w,\overline{w})}\right]}\right\rangle_{X}.
\end{equation}
\noindent
This formula provides an efficient framework for computing the ensemble average.
For this purpose one may employ the Fourier integral representation for the Gaussian regularization of the
$\delta-$distribution:
$\delta^{(2n)}\left(\left(X-z{\bf 1}\right){\bf v}\right)=
\lim_{\epsilon\to 0+}\delta^{(2n)}_{\epsilon}\left(\left(X-z{\bf 1}\right){\bf v}\right)$, where
\begin{equation}\label{Gaudelta_vec} 
  \delta^{(2n)}_{\epsilon}\left(\left(X-z{\bf 1}\right){\bf v}\right)
=\frac{1}{(2\pi\epsilon)^{n}}e^{-\frac{1}{2\epsilon } {\bf v}^* \left(X^*-\overline{z}{\bf 1}\right)
\left(X-z{\bf 1}\right)
{\bf v}}
=\int_{\mathbb{C}^n}e^{-\epsilon\frac{{\bf k}^*{\bf k}}{2}+\frac{i}{2}\left[{\bf k}^*\left(X-z{\bf 1}\right){\bf v}+
{\bf v}^*\left(X^*-\overline{z}{\bf 1}\right){\bf k}\right]}\frac{d{\bf k}d{\bf k}^*}{(2\pi)^{2n}}
\end{equation}
and we denoted  $d{\bf k}d{\bf k}^*:=\prod_{a=1}^nd \Re{k_a}d\Im{k_a}$ and ${\bf k}^*{\bf k}:=\sum_{a=1}^n\overline{k}_ak_a$ for ${\bf k}=(k_1,\ldots,k_n)^T\in \mathbb{C}^n$.
The Fourier integral representation of the $\delta-$distribution can be combined with the Berezin Gaussian integral representation for the
determinant:
\begin{align}\nonumber
& (-1)^{n} \det{Y_X(w,\overline{w})}=4^{n} \int \exp\left\{\frac{i}{2}\left(\mathbf{\Psi}_{+}^T,\mathbf{\Psi}_{-}^T\right)
  \left(\begin{array}{cc}0 & X-w{\bf 1}\\ X^*-\overline{w}{\bf 1}& 0\end{array}\right)\left(\begin{array}{c}\mathbf{\Phi}_{+} \\ \mathbf{\Phi}_{-}
  \end{array}\right)\right\}\,
  \mathcal {D}(\mathbf{\Psi},\mathbf{\Phi})\, \\
&\qquad \qquad =4^{n}\int \exp\left\{\frac{i}{2}\left(\mathbf{\Psi}_{-}^T (X^{*}-\overline{w})\mathbf{\Phi}_{+})+
\mathbf{\Psi}_{+}^T (X-w)\mathbf{\Phi}_{-}\right) \right\}\,
  \mathcal {D}(\mathbf{\Psi},\mathbf{\Phi})\label{Berezin_int}
  \end{align}
  where $\mathbf{\Psi}^T_{\sigma}=\left(\psi_{\sigma,1} ,\ldots,\psi_{\sigma,n} \right)$ and $\mathbf{\Phi}_{\sigma}=
  \left(\phi_{\sigma,1} ,\ldots,\phi_{\sigma,n} \right)^T$ for $\sigma=\pm$
  are vectors with anticommuting/Grassmann  entries,
  and we used the convention
  $\mathcal{D}(\mathbf{\Psi},\mathbf{\Phi}):=\mathcal{D}(\mathbf{\Psi}_{+},\mathbf{\Phi}_{+})
  \mathcal{D}(\mathbf{\Psi}_{-},\mathbf{\Phi}_{-})$ with
\[
\int F (\mathbf{\Psi}_{\sigma },\mathbf{\Phi}_{\sigma})\mathcal{D}(\mathbf{\Psi}_{\sigma },\mathbf{\Phi}_{\sigma }):=
\prod_{j=1}^n \partial_{\psi_{\sigma,j}} \partial_{\phi_{\sigma,j} } F (\mathbf{\Psi}_{\sigma },\mathbf{\Phi}_{\sigma })
,\quad \sigma=\pm.
\]
Note that $\mathcal{D}(\mathbf{\Psi}_{+},\mathbf{\Phi}_{+})
  \mathcal{D}(\mathbf{\Psi}_{-},\mathbf{\Phi}_{-})= (-1)^{n}\mathcal{D}(\mathbf{\Psi}_{+},\mathbf{\Phi}_{-})
  \mathcal{D}(\mathbf{\Psi}_{-},\mathbf{\Phi}_{+}).$
  
Exchanging the integration order (which can be done by our moment assumption), the average over $X$ boils down to computing the
Fourier transform
\[
\left\langle e^{\frac{i}{2}\mbox{\small Tr }\left(XM+X^{*}M'\right)}\right\rangle_{X}
\]
where $M,M'$ are $n\times n$ matrices whose components are even elements in the (complex) Grassmann algebra generated by
$\{\psi_{\pm,1},\dotsc ,\psi_{\pm,n},\phi_{\pm,1},\dotsc, \phi_{\pm,n} \}.$
In the next section we will apply this strategy to our model of two coupled Ginibre matrices.

\subsection{Implementation of the strategy for two coupled Ginibre matrices}
\label{sect:2.2}

The main goal of the present section is to prove the following
\begin{theorem}\label{Thm2.1}
Let $X$ be the random matrix defined in Eq.(\ref{ensemble_def}). For  $N\ge 2$  the averaged joint empirical density of
its eigenvalue $z$ at the origin $z=0$ and the associated pair of vectors ${\bf v}_{1}$ and ${\bf v}_2$, with respect to the
invariant measure on the sphere
\[
d\mu_N({\bf v}):=d\mu_H({\bf v})/C_{2N}=\delta\left({\bf v}^*_1{\bf v}_1+{\bf v}^*_2{\bf v}_2-1\right)\,d{\bf v}_1d{\bf v}^*_1 d{\bf v}_2d{\bf v}^*_2
\]
is given by
\begin{align} 
&
\mathbb{E}\left[\mathcal{P}_{2N}(z=0,  {\bf v}= ( {\bf v}_1,{\bf v}_2)\right]=
 \frac{1}{\pi^{2N+1}} \frac{e^{-|c|^2\left(\frac{{\bf v}_1^*{\bf v}_1}{{\bf v}_2^*{\bf v}_2}+\frac{{\bf v}_2^*{\bf v}_2}{{\bf v}_1^*{\bf v}_1}\right)}}
{\left({\bf v}_1^*{\bf v}_1\right)^N\left({\bf v}_2^*{\bf v}_2\right)^N}
\label{final_finite_size}\\
&
\qquad \qquad \times
\int_{\mathbb{C}^{2}} e^{-q_1\overline{q}_1-q_2\overline{q}_2}\left(q_1\overline{q}_2+|c|^2\right)^{N-2}\left(q_2\overline{q}_1+|c|^2\right)^{N-2}\,
\mathcal{D}(q_1,q_2,\overline{q}_1,\overline{q}_2)\,
\frac{dq_1d\overline{q}_1dq_2d\overline{q}_2}{\pi^2},\nonumber 
\end{align}   
where $\mathcal{D}(q_1,q_2,\overline{q}_1,\overline{q}_2)$ is given by the following expression:  
\begin{align}\label{parallel_explicit}
&\mathcal{D}(q_1,q_2,\overline{q}_1,\overline{q}_2)=|q_1|^2|q_2|^2\left(c^2\nu_1\overline{\nu}_2+\overline{c}^2\nu_2\overline{\nu}_1\right)\\
&\qquad\qquad +|c|^2\left[\left(1-\nu_1\nu_2\right)\left(q_1\overline{q}_2|q_1|^2+q_2\overline{q}_1|q_2|^2\right)+
|q_1|^2|q_2|^2\left(\frac{\overline{\nu}_1}{\nu_2}+\frac{\overline{\nu}_2}{\nu_1}\right)
\right]\nonumber\\
&\qquad\qquad
+|c|^2\left(c^2\nu_1\overline{\nu}_2+\overline{c}^2\nu_2\overline{\nu}_1\right)(q_1\overline{q}_2+q_2\overline{q}_1)\nonumber\\
&\qquad\qquad +|c|^4\left[\left(1-\nu_1\nu_2\right)(|q_1|^2+|q_2|^2)+q_1\overline{q}_2\left(\nu_1\overline{\nu}_1+\frac{\overline{\nu}_2}{\nu_1}\right)+
q_2\overline{q}_1\left(\nu_2\overline{\nu}_2+\frac{\overline{\nu}_1}{\nu_2}\right)
\right]\nonumber\\
&\qquad \qquad
+|c|^4\left(c^2\nu_1\overline{\nu}_2+\overline{c}^2\nu_2\overline{\nu}_1\right)(q_1\overline{q}_2+q_2\overline{q}_1)+|c|^6\left(\nu_1\overline{\nu}_1+\nu_2\overline{\nu}_2\right)\nonumber
\end{align}
 and we defined
   \begin{equation}\label{nu}
   \nu_1:=\frac{{\bf v}^*_1{\bf v}_2}{{\bf v}^*_1{\bf v}_1}, \quad \nu_2:=\frac{{\bf v}^*_2{\bf v}_1}{{\bf v}^*_2{\bf v}_2},
   \end{equation}
   so that $\overline{\nu}_1/\nu_2=\nu_1/\overline{\nu}_2=\frac{{\bf v}^*_2{\bf v}_2}{{\bf v}^*_1{\bf v}_1}$.
\end{theorem}
Note  that for $|c|>0$  both norms $p_1={\bf v}_1^*{\bf v}_1$ and $p_2={\bf v}_2^*{\bf v}_2$
must be simultaneously nonvanishing. Away from the measure-zero orthogonality set ${\bf v}_1^*{\bf v}_2=0$, both $\nu_1,\nu_2$ and $\overline{\nu}_1/\nu_2$ are well defined; the final expression extends to that set by continuity. \medskip

The above theorem further implies the following  result.
\begin{corollary}\label{Cor2.2}
Let $\Pi_N(z,p_1,p_2)$ be defined as in Eq.(\ref{def:Pi}). Then
the joint density $\mathbb{E}[\Pi_N(0, p_1,p_2)]$ (with respect to the measure $dp_1dp_2$)  of the
squared norms $p_1={\bf v}_1^*{\bf v}_1$ and $p_2={\bf v}_2^*{\bf v}_2$ corresponding to those eigenvectors
${\bf v} (a)$ of $X$ whose eigenvalues $z_a$ are in the vicinity of the  point $z=0$ takes the form
   $\mathbb{E}[\Pi_N(0, p_1,p_2)]=\overline{\Pi}_{N} (0,p_{1},p_{2})\delta (p_{1}+p_{2}-1) $ with 
\begin{align}
&\overline{\Pi}_{N} (0,p_{1},p_{2})=\frac{1}{\pi(N-1)!(N-2)!}\,\frac{e^{-|c|^2\left(\frac{p_1}{p_2}+\frac{p_2}{p_1}\right)}}
{p_1p_2} \label{final_finite_size1}\\
&\quad \times \int_{\mathbb{C}^{2}}  e^{-q_1\overline{q}_1-q_2\overline{q}_2}\left(q_1\overline{q}_2+|c|^2\right)^{N-2}\left(q_2\overline{q}_1+|c|^2\right)^{N-2}\,
\mathcal{U}(|q_1|^{2},|q_2|^{2},\overline{q}_1q_{2},\overline{q}_2q_{1})\,
\frac{dq_1d\overline{q}_1dq_2d\overline{q}_2}{\pi^2},\nonumber
\end{align}
where 
$\mathcal{U}$ is given by the following expression:
\[
\mathcal{U}(|q_1|^{2},|q_2|^{2},\overline{q}_1q_{2},\overline{q}_2q_{1})=
|c|^2\left[\frac{1}{N}\left(q_1|q_1|^2\overline{q}_2+q_2|q_2|^2\overline{q}_1\right)+\frac{1}{N-1}
|q_1|^2|q_2|^2\left(\frac{p_1}{p_2}+\frac{p_2}{p_1}\right)\right]
\]
\[
+|c|^4\left\{q_1\overline{q}_2\left[\frac{1}{N-1}\left(\frac{p_1}{p_2}+\frac{p_2}{p_1}\right)-\frac{1}{N}\frac{p_2}{p_1}\right]+
q_2\overline{q}_1\left[\frac{1}{N-1}\left(\frac{p_1}{p_2}+\frac{p_2}{p_1}\right)-\frac{1}{N}\frac{p_1}{p_2}\right]
\right\}
\]
\begin{equation}\label{parallel_explicit1}
+\frac{1}{N}|c|^4(|q_1|^2+|q_2|^2)+|c|^6\frac{1}{N(N-1)}\left(\frac{p_1}{p_2}+\frac{p_2}{p_1}\right).
\end{equation}
\end{corollary}
The rest of the section is devoted to the proof of these two results.

\begin{proof}[Proof of Theorem \ref{Thm2.1}]
Recall equations \eqref{main_practical}  \eqref{Gaudelta_vec}  \eqref{Berezin_int}.
The block structure of the matrix $X$ defined in Eq.(\ref{ensemble_def}) makes it natural to introduce
the following vectors:
\begin{equation}\label{set of vectors}
{\bf k}=\left(\begin{array}{c}{\bf k}_1\\{\bf k}_2\end{array}\right),\quad  \mathbf{\Phi}_{\pm}=\left(\begin{array}{c} \mathbf{\Phi}_{\pm,1}\\ \mathbf{\Phi}_{\pm,2}\end{array}\right), \,  \quad \mathbf{\Psi}^t_{\pm}=
\left(\begin{array}{c} \mathbf{\Psi}^t_{\pm,1}, \, \mathbf{\Psi}^t_{\pm,2}\end{array}\right),
\end{equation}
where  ${\bf k}_{\sigma}\in \mathbb{C}^N$ for $\sigma=1,2$ and each of the vectors
$\mathbf{\Phi}_{\pm,\sigma}$ and $\mathbf{\Psi}^t_{\pm,\sigma}$ has $N$ different anticommuting/Grassmann
entries. Using these definitions the equation (\ref{Gaudelta_vec}) can be written as
\begin{align}\label{Gaudelta_vec_block}
 \delta^{(4N)}_{\epsilon}\left(\left(X-z{\bf 1}\right){\bf v}\right)
&= \tfrac{1}{\pi^{2N}}  \int_{\mathbb{C}^{2N}}
\,\,e^{-\frac{\epsilon}{2}\left({\bf k}_1^*{\bf k}_1+{\bf k}_2^*{\bf k}_2\right)}e^{\frac{i}{2}  \left ({\bf k}_1^* (G_{1}-z){\bf v}_1+ {\bf v}_1^* (G_{1}^{*}-\overline{z} ){\bf k}_1\right)}\\
&
\times 
 e^{\frac{i}{2}  \left ( {\bf k}_2^* (G_{2}-z){\bf v}_2+ {\bf v}_2^* (G_{2}^{*}-\overline{z} ){\bf k}_2\right)}
 e^{\frac{i}{2}  \left (c ( {\bf k}_1^* {\bf v}_2+ {\bf v}_1^* {\bf k}_2)+ \overline{c} ( {\bf k}_2^* {\bf v}_1+ {\bf v}_2^* {\bf k}_1)
\right)} \tfrac{d{\bf k}_1d{\bf k}_1^*}{(4\pi)^{N}}\,\tfrac{d{\bf k}_2d{\bf k}_2^*}{(4\pi)^{N}}\nonumber
\end{align}
whereas Eq.(\ref{Berezin_int}) takes the form 
\begin{align}\label{Berezin_int_block}
 \det{ Y_X(w,\overline{w})}&=4^{2N}\int
e^{\frac{i}{2} \left( \mathbf{\Psi}_{+,1}^t  (\mathbf{G}_{1}-w) \mathbf{\Phi}_{-,1}+
\mathbf{\Psi}_{+,2}^t  (\mathbf{G}_{2}-w) \mathbf{\Phi}_{-,2}+ c\mathbf{\Psi}_{+,1}^t \mathbf{\Phi}_{-,2}+
\overline{c}  \mathbf{\Psi}_{+,2}^t\mathbf{\Phi}_{-,1}
 \right)} {\scriptstyle \mathcal {D}(\mathbf{\Psi}_{+},\mathbf{\Phi}_{-})}\\
 & 
\times \int   e^{\frac{i}{2} \left(  \mathbf{\Psi}_{-,1}^t  (\mathbf{G}^*_1-\overline{w}) \mathbf{\Phi}_{+,1}+
 \mathbf{\Psi}_{-,2}^t  (\mathbf{G}^*_2-\overline{w}) \mathbf{\Phi}_{+,2}+c\mathbf{\Psi}_{-,1}^t \mathbf{\Phi}_{+,2}+
\overline{c}  \mathbf{\Psi}_{-,2}^t\mathbf{\Phi}_{+,1} \right)} \,
 {\scriptstyle  \mathcal {D}(\mathbf{\Psi}_{-},\mathbf{\Phi}_{+})},
\nonumber
  \end{align}
  where 
\[
  \mathcal {D}(\mathbf{\Psi}_{+},\mathbf{\Phi}_{-}) \mathcal {D}(\mathbf{\Psi}_{-},\mathbf{\Phi}_{+})=
  \mathcal {D}(\mathbf{\Psi}_{+,1},\mathbf{\Phi}_{-,1})  \mathcal {D}(\mathbf{\Psi}_{-,1},\mathbf{\Phi}_{+,1}) \mathcal {D}(\mathbf{\Psi}_{+,2},\mathbf{\Phi}_{-,2})
\mathcal {D}(\mathbf{\Psi}_{-,2},\mathbf{\Phi}_{+,2}).
\]
Putting all this together we obtain   
\[
\mathbb{E}\left[\mathcal{P}_{2N}(z, {\bf v}_1,{\bf v}_2\right]=
\frac{1}{\pi}\lim_{|w-z|\to 0}\frac{\partial^2}{\partial w\partial \overline{w}}
\ \lim_{\epsilon \to 0} \mathbb{E}\left[ \delta^{(4N)}_{\epsilon }\left(\left(z{\bf 1}-X\right){\bf v}\right)\,
\det{ Y_X(w,\overline{w})} \right],
\]
where
\begin{align}\label{aveperformedA}
&\mathbb{E}\left[ \delta^{(4N)}_{\epsilon }\left(\left(z{\bf 1}-X\right){\bf v}\right)\,
\det{Y_X(w,\overline{w})} \right]=\\
&\tfrac{4^{2N}}{\pi^{2N}}   \int_{\mathbb{C}^{2N}} e^{-\frac{\epsilon}{2}\left({\bf k}_1^*{\bf k}_1+{\bf k}_2^*{\bf k}_2\right)}
e^{-\frac{i}{2} \left\{   z\left({\bf k}_1^*{\bf v}_1+{\bf k}_2^*{\bf v}_2\right)+\overline{z}
\left({\bf v}_1^*{\bf k}_1+{\bf v}_2^*{\bf k}_2\right) \right\} }\times e^{\frac{i}{2} \left\{ c\left({\bf k}_1^*{\bf v}_2+{\bf v}_1^*{\bf k}_2\right)+\overline{c}
\left({\bf k}_2^*{\bf v}_1+{\bf v}_2^*{\bf k}_1\right)  \right\}  }\nonumber\\
&
\hspace{-0,2cm} \int\hspace{-0,1cm}
e^{-\frac{i}{2}   \left\{  w\left(\mathbf{ \Psi}_{+,1}^t\mathbf{\Phi}_{-,1}+\mathbf{ \Psi}_{+,2}^t\mathbf{\Phi}_{-,2}\right)
+\overline{w} 
\left(\mathbf{ \Psi}_{-,1}^t\mathbf{\Phi}_{+,1}+\mathbf{ \Psi}_{-,2}^t\mathbf{\Phi}_{+,2}\right)\right\}  }
e^{\frac{i}{2} \left\{  c\left(\mathbf{ \Psi}_{+,1}^t\mathbf{\Phi}_{-,2}+\mathbf{ \Psi}_{-,1}^t\mathbf{\Phi}_{+,2}\right)+\overline{c}
\left(\mathbf{ \Psi}_{+,2}^t\mathbf{\Phi}_{-,1}+\mathbf{ \Psi}_{-,2}^t\mathbf{\Phi}_{+,1}\right) \right\}  }\nonumber\\
&\quad \prod_{j=1,2} \mathbb{E}_{G_j}\left[
e^{\frac{i}{2}\mbox{\small Tr }\left(G_j\left({\bf v}_j\otimes {\bf k}^*_j-\mathbf{ \Phi}_{-,j}\otimes \mathbf{\Psi}^t_{+,j}\right)
+G_j^*\left({\bf k}_j\otimes {\bf v}^*_j-\mathbf{ \Phi}_{+,j}\otimes \mathbf{\Psi}^t_{-,j}\right)\right)}
\right]
{\scriptstyle \mathcal {D}(\mathbf{\Psi}_{+},\mathbf{\Phi}_{-})\,\mathcal {D}(\mathbf{\Psi}_{-},\mathbf{\Phi}_{+})}
  \tfrac{d{\bf k}_1d{\bf k}_1^*}{(4\pi)^{N}}\,\tfrac{d{\bf k}_2d{\bf k}_2^*}{(4\pi)^{N}}.
\nonumber
\end{align}
Note that the average can be moved inside the integral since the random matrix elements are normal
distributed.  Hence the averaging amounts to using the following identity:
\begin{equation}\label{avegin}
\left\langle e^{\frac{i}{2}\mbox{\small Tr }\left(GB+G^*A\right)}\right\rangle_{G\in GinUE}
=e^{-\frac{1}{4}\mbox{\small Tr }\left(BA\right)}.
\end{equation}
Applying now the rule Eq.(\ref{avegin}) to evaluating the two averages above we get for $j=1,2,$
\begin{align}
& \mathbb{E}_{G_j}\left[ e^{\frac{i}{2}\mbox{\small Tr }\left(G_j\left({\bf v}_j\otimes {\bf k}^*_j-
\mathbf{ \Phi}_{-,j}\otimes \mathbf{\Psi}^t_{+,j}\right)
+G_j^*\left({\bf k}_j\otimes {\bf v}^*_j-\mathbf{ \Phi}_{+,j}\otimes \mathbf{\Psi}^t_{-,j}\right)\right)}
\right]\\
&=
e^{-\frac{1}{4} \left\{  \left( {\bf k}_{j}^*{\bf k}_{j}\right) \left({\bf v}_{j}^*{\bf v}_{j}\right)-
\left( {\bf k}_{j}^*\mathbf{\Phi}_{+,j}\right)\left(\mathbf{\Psi}_{-,j}^t{\bf v}_{j}\right)-
\left({\bf v}_{j}^*\mathbf{\Phi}_{-,j}\right)\left(\mathbf{\Psi}_{+,j}^t{\bf k}_{j}\right)-
\left(\mathbf{\Psi}^t_{+,j}\mathbf{\Phi}_{+,j}\right)\left(\mathbf{\Psi}_{-,j}^t\mathbf{\Phi}_{-,j}\right)
\right\}}.\nonumber
\end{align}
At this point the following observations are due. The integral over both vectors ${\bf k}_{j}, j=1,2$ remains (i) Gaussian after the Ginibre averaging 
and (ii) convergent even in the limit $\epsilon\to 0+$ due to the fact that
for $|c|>0$  both norms $p_1={\bf v}_1^*{\bf v}_1$ and $p_2={\bf v}_2^*{\bf v}_2$  must be simultaneously nonvanishing.
Therefore one can safely set $\epsilon=0$ in the integrand and perform the corresponding ${\bf k}$-integrations explicitly.
This can be done for any value of the spectral parameter $z$, but for $z\ne 0$
the resulting expression and subsequent manipulations become relatively cumbersome. Having in mind our main goal -
studying in detail the process of eigenvector ergodization due to coupling between subsystems -
we concentrate from this point on studying such an ergodization at the origin of the complex plane $z=0$ for the reasons
of simplicity and transparency of the presentation.

For $z=0$ and $\epsilon=0$ the Gaussian integral for ${\bf k}_1$ has the following form:
\begin{equation}\label{k1}
\int_{\mathbb{C}^{N}}
e^{-\frac{1}{4}{\bf k}_1^*{\bf k}_1({\bf v}_1^*{\bf v}_1)+\frac{i}{2}\left({\bf k}_1^*{\bf a}_1+{\bf b}_1^t{\bf k}_1\right)} \tfrac{d{\bf k}_1d{\bf k}_1^*}{(4\pi)^{N}}\,=\tfrac{1}{({\bf v}_1^*{\bf v}_1)^N}\,e^{-\frac{{\bf b}_1^t{\bf a}_1}{{\bf v}_1^*{\bf v}_1}},
\end{equation}
where ${\bf a}_1:=c{\bf v}_2-\frac{i}{2}\mathbf{\Phi}_{+,1}\left(\mathbf{\Psi}^t_{-,1}{\bf v}_1\right),$
${\bf b}^t_1:=\overline{c}{\bf v}^*_2-\frac{i}{2}\left({\bf v}_1^*\mathbf{\Phi}_{-,1}\right)\mathbf{\Psi}^t_{+,1}$ and
\begin{align*}
{\scriptstyle {\bf b}_1^t{\bf a}_1}&{\scriptstyle = |c|^{2} {\bf v}^*_2{\bf v}_2- \tfrac{i}{2}c \left({\bf v}_1^*\mathbf{\Phi}_{-,1}\right)
 \left (\mathbf{\Psi}^t_{+,1} {\bf v}_2\right)-  \tfrac{i}{2}\overline{c}\left({\bf v}^*_2\mathbf{\Phi}_{+,1}\right)
 \left(\mathbf{\Psi}^t_{-,1}{\bf v}_1\right)- \tfrac{1}{4} \left({\bf v}_1^*\mathbf{\Phi}_{-,1}\right)
  \left(\mathbf{\Psi}^t_{+,1}\mathbf{\Phi}_{+,1}\right)\left(\mathbf{\Psi}^t_{-,1}{\bf v}_1\right)}\\
&{\scriptstyle = |c|^{2} {\bf v}^*_2{\bf v}_2- \tfrac{i}{2}\left( c \left({\bf v}_1^*\mathbf{\Phi}_{-,1}\right)
 \left (\mathbf{\Psi}^t_{+,1} {\bf v}_2\right) + \overline{c}\left({\bf v}^*_2\mathbf{\Phi}_{+,1}\right)
 \left(\mathbf{\Psi}^t_{-,1}{\bf v}_1\right) \right)+ \tfrac{1}{4} \left(\mathbf{\Psi}^t_{+,1}\mathbf{\Phi}_{+,1}\right)
  \left(\mathbf{\Psi}^t_{-,1}{\bf v}_1\otimes {\bf v}_1^*\mathbf{\Phi}_{-,1}\right).}
\end{align*}
Similarly,  the  Gaussian integral for ${\bf k}_2$ has the form
\begin{equation}\label{k2}
\int_{\mathbb{C}^{N}}
e^{-\frac{1}{4}{\bf k}_2^*{\bf k}_2({\bf v}_2^*{\bf v}_2)+\frac{i}{2}\left({\bf k}_2^*{\bf a}_2+{\bf b}_2^t{\bf k}_2\right)} \tfrac{d{\bf k}_2d{\bf k}_2^*}{(4\pi)^{N}}\, =
\tfrac{1}{({\bf v}_2^*{\bf v}_2)^N}\,e^{-\frac{{\bf b}_2^t{\bf a}_2}{{\bf v}_2^*{\bf v}_2}},
\end{equation}
where ${\bf a}_2:=\overline{c}{\bf v}_1- \frac{i}{2}\mathbf{\Phi}_{+,2}\left(\mathbf{\Psi}^t_{-,2}{\bf v}_2\right),$ $ {\bf b}^t_2:=c{\bf v}^*_1- \frac{i}{2}\left({\bf v}_2^*\mathbf{\Phi}_{-,2}\right)\mathbf{\Psi}^t_{+,2},$
and
\begin{align*}
{\scriptstyle{\bf b}_2^t{\bf a}_2}&={\scriptstyle |c|^{2} {\bf v}^*_1{\bf v}_1-  \tfrac{i}{2}c \left({\bf v}_1^*\mathbf{\Phi}_{+,2}\right)
 \left (\mathbf{\Psi}^t_{-,2} {\bf v}_2\right)-   \tfrac{i}{2}\overline{c}\left({\bf v}^*_2\mathbf{\Phi}_{-,2}\right)
 \left(\mathbf{\Psi}^t_{+,2}{\bf v}_1\right)- \tfrac{1}{4} \left({\bf v}_2^*\mathbf{\Phi}_{-,2}\right)
  \left(\mathbf{\Psi}^t_{+,2}\mathbf{\Phi}_{+,2}\right)\left(\mathbf{\Psi}^t_{-,2}{\bf v}_2\right)}\\
&{\scriptstyle =  |c|^{2} {\bf v}^*_1{\bf v}_1-  \tfrac{i}{2}\left(c
\left({\bf v}_1^*\mathbf{\Phi}_{+,2}\right)
 \left (\mathbf{\Psi}^t_{-,2} {\bf v}_2\right)+\overline{c}\left({\bf v}^*_2\mathbf{\Phi}_{-,2}\right)
 \left(\mathbf{\Psi}^t_{+,2}{\bf v}_1\right)
\right)+ \tfrac{1}{4}  \left(\mathbf{\Psi}^t_{+,2}\mathbf{\Phi}_{+,2}\right)
 \left(\mathbf{\Psi}^t_{-,2}{\bf v}_2\otimes {\bf v}_2^*\mathbf{\Phi}_{-,2}\right).}
\end{align*}
Substituting those identities back to Eq.(\ref{aveperformedA}) we get
\begin{align}\label{aveperformedC}
&\mathbb{E}\left[ \delta^{(4N)}\left(\left(z{\bf 1}-X\right){\bf v}\right)\,
\det{Y_X(w,\overline{w})} \right]=\tfrac{4^{2N}}{\pi^{2N}}\tfrac{e^{-|c|^2\left(\frac{{\bf v}_1^*{\bf v}_1}{{\bf v}_2^*{\bf v}_2}+\frac{{\bf v}_2^*{\bf v}_2}{{\bf v}_1^*{\bf v}_1}\right)}}
{\left({\bf v}_1^*{\bf v}_1\right)^N\left({\bf v}_2^*{\bf v}_2\right)^N}\\
& \int \hspace{-0,1cm} e^{-\frac{i}{2}   \left\{  w\left(\mathbf{ \Psi}_{+,1}^t\mathbf{\Phi}_{-,1}+\mathbf{ \Psi}_{+,2}^t\mathbf{\Phi}_{-,2}\right)
+\overline{w} 
\left(\mathbf{ \Psi}_{-,1}^t\mathbf{\Phi}_{+,1}+\mathbf{ \Psi}_{-,2}^t\mathbf{\Phi}_{+,2}\right)\right\}  }
e^{\frac{i}{2} \left\{  c\left(\mathbf{ \Psi}_{+,1}^t\mathbf{\Phi}_{-,2}+\mathbf{ \Psi}_{-,1}^t\mathbf{\Phi}_{+,2}\right)+\overline{c}
\left(\mathbf{ \Psi}_{+,2}^t\mathbf{\Phi}_{-,1}+\mathbf{ \Psi}_{-,2}^t\mathbf{\Phi}_{+,1}\right) \right\}  }\nonumber\\
& \quad e^{- \tfrac{i}{2 {\bf v}_1^*{\bf v}_1}
  \left( c \left (\mathbf{\Psi}^t_{+,1} {\bf v}_2\right)\left( {\bf v}_1^*\mathbf{\Phi}_{-,1}\right) +
\overline{c} \left(\mathbf{\Psi}^t_{-,1}{\bf v}_1\right) \left({\bf v}^*_2\mathbf{\Phi}_{+,1}\right) \right)
-\tfrac{i}{2 {\bf v}_2^*{\bf v}_2} \left( c\left (\mathbf{\Psi}^t_{-,2} {\bf v}_2\right)
 \left({\bf v}_1^*\mathbf{\Phi}_{+,2}\right)+\overline{c}
 \left(\mathbf{\Psi}^t_{+,2}{\bf v}_1\right)\left({\bf v}^*_2\mathbf{\Phi}_{-,2}\right)
\right)}\nonumber\\
&\quad e^{\frac{1}{4}  \left(\mathbf{\Psi}_{+,1}^t \mathbf{\Phi}_{+,1}\right) \left(\mathbf{\Psi}_{-,1}^t\left[{\bf 1}_{N}-
\frac{{\bf v}_{1}\otimes {\bf v}_1^*}{{\bf v}^*_1{\bf v}_1}\right]\mathbf{\Phi}_{-,1}\right)+\frac{1}{4}\left(\mathbf{\Psi}_{+,2}^t \mathbf{\Phi}_{+,2}\right) \left(\mathbf{\Psi}_{-,2}^t\left[{\bf 1}_{N}-\frac{{\bf v}_2\otimes {\bf v}_2^*}{{\bf v}^*_2{\bf v}_2}\right]\mathbf{\Phi}_{-,2}\right) } {\scriptstyle \mathcal {D}(\mathbf{\Psi}_{+},\mathbf{\Phi}_{-})\,\mathcal {D}(\mathbf{\Psi}_{-},\mathbf{\Phi}_{+})}
\nonumber 
\end{align}
\vspace{0.5ex}

The last line in the integrand above contains the exponential of two terms {\it quartic} in the  anticommuting components, making the whole integral non-Gaussian.
A way to progress further is to employ the so-called
 Hubbard-Stratonovich transformation, that is by paying the price of auxilliary Gaussian integrals to trade those quartic terms for terms {\it quadratic} in the anticommuting variables.
Doing this amounts to exploiting the following identities:
\begin{align}\label{HS1}
& e^{\frac{1}{4}  \left(\mathbf{\Psi}_{+,1}^* \mathbf{\Phi}_{+,1}\right)
\left(\mathbf{\Psi}_{-,1}^t\left[{\bf 1}_{N}-\frac{{\bf v}_{1}\otimes {\bf v}_1^*}{{\bf v}^*_1{\bf v}_1}\right]\mathbf{\Phi}_{-,1}\right)}=
\int_{\mathbb{C}} e^{-|q_{1}|^{2}}e^{-\frac{q_1}{2}\left(\mathbf{\Psi}_{+,1}^t \mathbf{\Phi}_{+,1}\right)
-\frac{\overline{q}_1}{2}\left(\mathbf{\Psi}_{-,1}^t
\left[{\bf 1}_{N}-\frac{{\bf v}_{1}\otimes {\bf v}_1^*}{{\bf v}_1^*{\bf v}_1}\right]\mathbf{\Phi}_{-,1}\right)}
\tfrac{dq_1d\overline{q}_1}{\pi}\\
& e^{\frac{1}{4} \left(\mathbf{\Psi}_{+,2}^t \mathbf{\Phi}_{+,2}\right)
\left(\mathbf{\Psi}_{-,2}^t\left[{\bf 1}_{N}-\frac{{\bf v}_2\otimes {\bf v}_2^*}{{\bf v}^*_2{\bf v}_2}\right]\mathbf{\Phi}_{-,2}\right) }=
\int_{\mathbb{C}} e^{-|q_{2}|^{2}}e^{-\frac{q_2}{2}   \left(\mathbf{\Psi}_{+,2}^t \mathbf{\Phi}_{+,2}\right)
-\frac{\overline{q}_2}{2}
\left(\mathbf{\Psi}_{-,2}^t\left[{\bf 1}_{N}-\frac{{\bf v}_2\otimes {\bf v}_2^*}{{\bf v}^*_2{\bf v}_2}\right]\mathbf{\Phi}_{-,2}\right)}
\tfrac{dq_2d\overline{q}_2}{\pi}\label{HS2}
\end{align}
Substituting these two integrals back to Eq.(\ref{aveperformedC}) and changing the order of integration one is then
immediately able to perform
 the (by now, Gaussian) integrals over all anticommuting variables and get finally:
\begin{align}\label{intermediate}
&\mathbb{E}\left[\delta^{(4N)}\left(\left(z{\bf 1}-X\right){\bf v}\right)\,
\det{Y_X(w,\overline{w})}\right]=\nonumber\\
&\qquad \qquad\qquad
 =\tfrac{4^{2N}}{\pi^{2N}}\tfrac{e^{-|c|^2\left(\frac{{\bf v}_1^*{\bf v}_1}{{\bf v}_2^*{\bf v}_2}+\frac{{\bf v}_2^*{\bf v}_2}{{\bf v}_1^*{\bf v}_1}\right)}}
{({\bf v}_1^*{\bf v}_1)^N({\bf v}_2^*{\bf v}_2)^N}\,\int e^{-q_1\overline{q}_1-q_2\overline{q}_2}
\int
e^{-\frac{1}{2} \mathbf{\Psi}^t D_N \mathbf{\Phi}}
{\scriptstyle \mathcal {D}(\mathbf{\Psi}_{+},\mathbf{\Phi}_{-})\,\mathcal {D}(\mathbf{\Psi}_{-},\mathbf{\Phi}_{+})}\nonumber\\
&\qquad \qquad \qquad= \tfrac{4^{2N}}{\pi^{2N}}\tfrac{1}{2^{4N}}\tfrac{e^{-|c|^2\left(\frac{{\bf v}_1^*{\bf v}_1}{{\bf v}_2^*{\bf v}_2}+\frac{{\bf v}_2^*{\bf v}_2}{{\bf v}_1^*{\bf v}_1}\right)}}
{({\bf v}_1^*{\bf v}_1)^N({\bf v}_2^*{\bf v}_2)^N}\,\int e^{-q_1\overline{q}_1-q_2\overline{q}_2}\det{D_N}\,\,
\tfrac{dq_1d\overline{q}_1dq_2d\overline{q}_2}{\pi^2}
\end{align}
 where in the first line 
\[
\mathbf{\Psi}^t= (\mathbf{\Psi}_{+,1}^t,\mathbf{\Psi}_{-,1}^t,\mathbf{\Psi}_{+,2}^t,\mathbf{\Psi}_{-,2}^t),
\qquad \mathbf{\Phi}= (\mathbf{\Phi}_{-,1}^{t},\mathbf{\Phi}_{+,1}^{t}, \mathbf{\Phi}_{-,2}^{t},\mathbf{\Phi}_{+,2}^{t} )^{t},
\]
and $D_{N}$ is the $4N\times 4N$ matrix
\begin{align}\label{maindet}
D_{N}&=D_{N} (w,\bar w)= D_N(w,\bar w;q_1,q_2,  \bar q_{1} ,\bar q_{2})=\\
&\quad
\begin{pmatrix}
i \left (w{\bf 1}_N +c \tfrac{{\bf v}_2\otimes {\bf v}_1^*}{{\bf v}_1^*{\bf v}_1}\right ) & q_{1}{\bf 1}_N  & -ic{\bf 1}_N  & 0 \\
\bar q_{1}\left ({\bf 1}_N -\tfrac{{\bf v}_{1}\otimes {\bf v}_1^*}{{\bf v}_1^*{\bf v}_1} \right)&
i\left  ( \bar w{\bf 1}_N +\bar c \tfrac{{\bf v}_1\otimes {\bf v}_2^*}{{\bf v}_1^*{\bf v}_1} \right) & 0& -ic{\bf 1}_N \\
-i \bar c{\bf 1}_N  & 0 &  i \left( w{\bf 1}_N +\bar c \tfrac{{\bf v}_1\otimes {\bf v}_2^*}{{\bf v}_2^*{\bf v}_2} \right)&  q_{2}{\bf 1}_N  \\
0 & -i \bar c{\bf 1}_N  & \bar q_{2} \left( {\bf 1}_N -\tfrac{{\bf v}_2\otimes {\bf v}_2^*}{{\bf v}_2^*{\bf v}_2}\right ) & i \left ( \bar w{\bf 1}_N + c\tfrac{{\bf v}_2\otimes {\bf v}_1^*}{{\bf v}_2^*{\bf v}_2}\right )\\
\end{pmatrix}
\nonumber
 \end{align}
The next task is to evaluate $\det D_N$. Since the off-diagonal blocks are multiples of the identity we can write
\[
\det D_N= \det  \begin{pmatrix}
X_{1} & -ic{\bf 1}_{2N} \\
-i \bar c  {\bf 1}_{2N}& X_{2}\\
\end{pmatrix}= \det (X_{1}X_{2}+|c|^{2}{\bf 1}_{2N})= \det M
\]
where $M$ is the $2N\times 2N$ block matrix 
\[
M= M(w,\bar w)=M(w,\bar w;q_1,q_2,  \bar q_{1} ,\bar q_{2})=\begin{pmatrix}
M_{11} & M_{12}\\
M_{21} & M_{22} \\
 \end{pmatrix},
\]
with $N\times N$ blocks
\begin{align*}
M_{11}&=-w^{2}{\bf 1}_{N}-w\left (c  \tfrac{{\bf v}_2\otimes {\bf v}_1^*}{{\bf v}_1^*{\bf v}_1}+\bar c\tfrac{{\bf v}_1\otimes {\bf v}_2^*}{{\bf v}_2^*{\bf v}_2} \right)+
(|c|^{2}+  q_{1} \bar q_{2}  ) \left( {\bf 1}_{N}-\tfrac{{\bf v}_2\otimes {\bf v}_2^*}{{\bf v}_2^*{\bf v}_2} \right )\\
M_{22}&=-\bar w^{2}{\bf 1}_{N}-\bar w\left (  c \tfrac{{\bf v}_2\otimes {\bf v}_1^*}{{\bf v}_2^*{\bf v}_2}+\bar c \tfrac{{\bf v}_1\otimes {\bf v}_2^*}{{\bf v}_1^*{\bf v}_1}\right) +
 (|c|^{2}+ \bar q_{1}  q_{2}  ) \left({\bf 1}_{N} -\tfrac{{\bf v}_1\otimes {\bf v}_1^*}{{\bf v}_1^*{\bf v}_1}\right )\\
M_{12}&=iw  q_{2}{\bf 1}_{N}+ i \bar w  q_{1}{\bf 1}_{N}+ic\left ( q_{2} \tfrac{{\bf v}_2\otimes {\bf v}_1^*}{{\bf v}_1^*{\bf v}_1}+q_{1}\tfrac{{\bf v}_2\otimes {\bf v}_1^*}{{\bf v}_2^*{\bf v}_2}\right )\\
M_{21}&=wi\bar q_{1}\left ({\bf 1}_{N}- \tfrac{{\bf v}_1\otimes {\bf v}_1^*}{{\bf v}_1^*{\bf v}_1}\right)+ \bar w i\bar q_{2}\left( {\bf 1}_{N}-\tfrac{{\bf v}_2\otimes {\bf v}_2^*}{{\bf v}_2^*{\bf v}_2} \right ).
\end{align*}
Remember the definition \eqref{nu} of $\nu_{1},\nu_{2}$ and set 
\[
\lambda :=  \tfrac{\|{\bf v}_{2}\|}{\|{\bf v}_{1}\|},\qquad {\bf u}_{1}:= \tfrac{ {\bf v}_{1}}{\| {\bf v}_{1}\|}  ,
\qquad  {\bf u}_{2}:= \tfrac{  ( {\bf v}_{2}-\nu_{1}{\bf v}_{1}) }{\| {\bf v}_{2}\|\sqrt{1-\nu_{1}\nu_{2} }} .
\]
Then $\|{\bf u}_1\|=\|{\bf u}_2\|=1$ and 
\[
{\bf u}_{2}^{*}{\bf u}_{1}={\bf u}_{2}^{*}{\bf v}_{1}=0={\bf u}_{1}^{*}{\bf u}_{2}={\bf u}_{1}^{*}{\bf u}_{2}.
\]
Let ${\bf u}_{3},\dots {\bf u}_{N}$ be an orthonormal basis of $(\mathrm{span}\{{\bf v}_{1},{\bf v}_{2} \})^{\perp}$.
In the basis $({\bf u}_{1},\dots {\bf u}_{N}),$ the matrix $M$ takes the block-diagonal form
\[
M= \begin{pmatrix}
M' & 0 \\
0 & M''\\
\end{pmatrix}
\]
where  $M''=P\otimes {\bf 1}_{N-2}, $ with
\[
P:= \begin{pmatrix}
-w^{2}+(|c|^{2}+  q_{1} \bar q_{2}  ) & i ( q_{2}w+  q_{1} \bar w) \\
i ( w\bar q_{1}+ \bar w \bar q_{2})& -\bar w^{2}+ (|c|^{2}+ \bar q_{1}  q_{2}  ) \\
\end{pmatrix}.
\]
Finally $M'$ is a linear combination of  $4\times 4$ matrices
\[
M' (w,\bar w)=M_{0} -w^{2}{\bf 1}_{4}  + w M_{1} -\bar w^{2} {\bf 1}_{4}+\bar w  M_{2} ,
\]
where $M_{0}=M' (0,0):= \begin{pmatrix}
A_{0} & C_{0}\\
0 & B_{0}\\
\end{pmatrix}$ with $2\times 2$ blocks
\begin{align*}
A_{0}&:= (|c|^{2}+  q_{1}\bar q_{2}) \alpha_{0}, \qquad
\alpha_{0}:=  \begin{pmatrix}
1-\nu_{1}\nu_{2}& -\lambda \bar \nu_{2}\sqrt{1-\nu_{1}\nu_{2}}\\
-\lambda  \nu_{2}\sqrt{1-\nu_{1}\nu_{2}}& \nu_{1}\nu_{2}
\end{pmatrix}\\
B_{0}&:= (|c|^{2}+ \bar q_{1} q_{2})  E_{22},\qquad
E_{22}:=  \begin{pmatrix}
0& 0\\
0& 1\\
\end{pmatrix}\\
C_{0}&:= ic\begin{pmatrix}
q_{2}\nu_{1}+q_{1}\bar \nu_{2}&0\\
( q_{2}\lambda+q_{1}\lambda^{-1}) \sqrt{1-\nu_{1}\nu_{2}}& 0\\
\end{pmatrix},
\end{align*}
\[
 M_{1} :=  \begin{pmatrix}
-\alpha_{1} & i q_{2}{\bf 1}_{2} \\
i\bar q_{1} E_{22} & 0\\
\end{pmatrix}, \qquad \alpha_{1}=   \begin{pmatrix}
c\nu_{1}+\bar c \nu_{2}&  \bar c \lambda^{-1} \sqrt{1-\nu_{1}\nu_{2}}\\
 c \lambda\sqrt{1-\nu_{1}\nu_{2}}& 0\\
\end{pmatrix}
\]
\[
M_{2}  := \begin{pmatrix}
0 & i q_{1}{\bf 1}_{2} \\
i\bar q_{2} \alpha_{0} & -\beta_{1}\\
\end{pmatrix}, \qquad  \beta_{1}=   \begin{pmatrix}
\bar c\bar\nu_{1}+c \bar\nu_{2}&  \bar c \lambda \sqrt{1-\nu_{1}\nu_{2}}\\
 c \lambda^{-1}\sqrt{1-\nu_{1}\nu_{2}}& 0\\
\end{pmatrix}.
\]
Putting all this together we obtain
\[
\det D (w,\bar w)=\det M' (w,\bar w) \ ( \det P (w,\bar w) )^{N-2}.
\]
We compute $\det M'(0,0)= \det A_{0}\det B_{0}=0$ and $\partial_{w}  \det P_{|w=\bar w=0}= \partial_{\bar w}  \det P_{|w=\bar w=0}= 0. $
It follows
\begin{align*}
\partial_{w}\partial_{\bar w} \det D_{|w=\bar w=0}&=   \det P (0,0 )^{N-2}  (\partial_{w} \partial_{\bar w} \det M')_{|w=\bar w=0}\\
&=
( q_{1}\bar q_{2}+|c|^{2})^{N-2} ( \bar q_{1} q_{2} +  |c|^{2})^{N-2}  (\partial_{w} \partial_{\bar w} \det M')_{|w=\bar w=0} .
\end{align*}
To compute the derivatives, one may diagonalize $M',$ or the equivalent $8\times 8$ matrix obtained restricting to the subspace
$\mathrm{span}\{{\bf v}_{1}, {\bf v}_{2}\},$  by solving the eigenvector equations and  compute the corresponding determinant 
using Mathematica software. The result can be schematically
written as follows:
   \begin{equation}\label{factor_parallel}
    \det M'(w,\overline{w};q_1,q_2,\overline{q}_1,\overline{q}_2)=w\overline{w}\left(\mathcal{F}(q_1,q_2,\overline{q}_1,\overline{q}_2)+
    O (|w|)\right)
\end{equation}
where $\mathcal{F}(q_1,q_2,\overline{q}_1,\overline{q}_2)$ is explicitly given in Eq.(\ref{parallel_explicit}).
We give here an analytic derivation of the formula. The starting point is a regularization of $M_{0}$ which makes the matrix invertible. Precisely we write
\[
 (\partial_{w} \partial_{\bar w} \det M')_{|w=\bar w=0} =
 \lim_{\varepsilon  \to 0+}   (\partial_{w} \partial_{\bar w} \det M'_{\varepsilon } (w,\bar w  ))_{|w=\bar w=0},
\]
with $M'_{\varepsilon } (w,\bar w ):= M_{0,\varepsilon } -w^{2}{\bf 1}_{4}  + w M_{1} -\bar w^{2} {\bf 1}_{4} +\bar w M_{2}$ and
\[
M_{0,\varepsilon} := \begin{pmatrix}
A_{0,\varepsilon }& C_{0}\\
0 & B_{0,\varepsilon }\\
\end{pmatrix},\qquad A_{0,\varepsilon }:= (|c|^{2}+  q_{1}\bar q_{2}) (\alpha_{0}+\hat{\varepsilon} ),
\quad  B_{0,\varepsilon }:= (|c|^{2}+ \bar q_{1} q_{2})  ( E_{22}+\hat{\varepsilon} ).
\]
where we introduced $\hat{\varepsilon}:=diag(\varepsilon,\varepsilon)$.
Hence we obtain
\begin{align*}
&(\partial_{w} \partial_{\bar w} \det M')_{|w=\bar w=0} =\\
&\quad 
 \lim_{\varepsilon  \to 0+} \left[\tr \big(M_{0,\varepsilon }^{-1} M_{1}\big)\, \tr \big(M_{0,\varepsilon }^{-1} M_{2}\big)\, \det M_{0,\varepsilon }\ -\ 
\tr \big(M_{0,\varepsilon }^{-1} M_{1}M_{0,\varepsilon }^{-1} M_{2} \big) \, \det M_{0,\varepsilon }
\right].
\end{align*}
The problem boils down to computing $M_{0,\varepsilon }^{-1}$ and the traces above in the limit $\varepsilon \to 0+.$\medskip 

Using the block structure of $M_{0,\varepsilon }$ and setting $d:= |c|^{2}+q_{1}\bar q_{2}$ we compute
\[
\det M_{0,\varepsilon }= (d \bar d )^{2} \varepsilon^{2} (1+\varepsilon )^{2},
\]
and
\[
M_{0,\varepsilon }^{-1}= \begin{pmatrix}
A_{0,\varepsilon }^{-1}& - A_{0,\varepsilon }^{-1}C_{0}B_{0,\varepsilon }^{-1}\\
0 & B_{0,\varepsilon }^{-1}\\
\end{pmatrix}=
\begin{pmatrix}
(d)^{-1}(\alpha_{0}+\hat{\varepsilon})^{-1}& -\frac{1}{\varepsilon^{2} } (d\bar d)^{-1} C_{0}\\
0 & (\bar d)^{-1} (E_{22}+\hat{\varepsilon} )^{-1}\\
\end{pmatrix},
\]
with
\[
(\alpha_{0}+\hat{\varepsilon} )^{-1}= \tfrac{1}{\varepsilon (1+\varepsilon )}
\begin{pmatrix}
\nu_{1}\nu_{2}+\varepsilon & \lambda \bar  \nu_{2}\sqrt{1-\nu_{1}\nu_{2}} \\
 \lambda  \nu_{2}\sqrt{1-\nu_{1}\nu_{2}}  & 1-\nu_{1}\nu_{2}+\varepsilon \\
\end{pmatrix},\quad
(E_{22}+\hat{\varepsilon} )^{-1}=\begin{pmatrix}
\frac{1}{\varepsilon }&0 \\
0  & \frac{1}{1+\varepsilon }\\
\end{pmatrix}.
\]
Inserting these formulas in the traces above we obtain
\begin{align*}
&\mathrm{tr} (M_{0,\varepsilon }^{-1}M_{1})= -\frac{1}{\varepsilon d} (\nu_{1}c+\nu_{2}\bar c),\quad
\mathrm{tr} (M_{0,\varepsilon }^{-1}M_{2})=- \frac{1}{\varepsilon \bar d} (\bar \nu_{1}\bar  c+\bar \nu_{2} c)\\
&\mathrm{tr} (M_{0,\varepsilon }^{-1}M_{1}M_{0,\varepsilon }^{-1}M_{2})=
-\frac{|c|^{2} (1-\nu_{1}\nu_{2}) }{\varepsilon^{2} (1+\varepsilon )(\bar d d) }
\left[\frac{\bar q_{2}}{d} (q_{2}+q_{1}\lambda^{-2}) + \frac{\bar q_{1}}{\bar d} (q_{2}\lambda^{2}+q_{1}) \right]+ O \left(\frac{1}{\varepsilon } \right).
\end{align*}
Hence 
\begin{align*}
& \lim_{\varepsilon  \to 0+}   (\partial_{w} \partial_{\bar w} \det M'_{\varepsilon } (w,\bar w  ))_{|w=\bar w=0}=
( |c|^{2}+q_{1}\bar q_{2})\ ( |c|^{2}+\bar q_{1} q_{2}) \  (\nu_{1}c+\nu_{2}\bar c) (\bar \nu_{1}\bar  c+\bar \nu_{2} c) \\
&\qquad
 +|c|^{2} (1-\nu_{1}\nu_{2})|q_{1}|^{2}|q_{2}|^{2}\left[ \left(\frac{1}{\lambda^{2}}+\lambda^{2} \right) +
\left(\frac{\bar q_{1}q_{2}}{|q_{1}|^{2}}+\frac{ q_{1}\bar q_{2}}{|q_{2}|^{2}} \right)
\right]\\
&\qquad
 +|c|^{4} (1-\nu_{1}\nu_{2})\left[ |q_{1}|^{2}+ |q_{2}|^{2}+\bar q_{1}q_{2} \lambda^{2}+\frac{ q_{1}\bar q_{2}}{\lambda^{2}} \right].
\end{align*}
The result now follows by inserting the definition for $\lambda$ and reorganizing the sum in terms of powers of $c.$

\end{proof}

\begin{proof}[Proof of  Corollary \ref{Cor2.2}]
We replace $\delta ({\bf v}_1^*{\bf v}_1+{\bf v}_2^*{\bf v}_2)-1)$ with $ \delta (p_{1}+p_{2}-1).$
Our goal is to compute, for $p_{2}=1-p_{1},$
\[
\overline{\Pi}_{N} (0,p_{1},p_{2})=\int_{\mathbb{C}^{2N}}  \mathbb{E}\left[\mathcal{P}_{2N}(0,   ( {\bf v}_1,{\bf v}_2)\right]
\delta(p_1-{\bf v}_1^*{\bf v}_1)\delta(p_2-{\bf v}_2^*{\bf v}_2)\ d{\bf v}_1^*d{\bf v}_1d{\bf v}_2^*d{\bf v}_2.
\]
 Inserting formulas \eqref{final_finite_size} and
\eqref{parallel_explicit}, the above integral becomes 
\begin{align}
&\overline{\Pi}_{N} (0,p_{1},p_{2})=  \frac{1}{\pi^{2N+1}}\tfrac{e^{-|c|^{2}  \left(\frac{p_{1}}{p_{2}}+\frac{p_{2}}{p_{1}} \right)}}{p_{1}^{N}p_{2}^{N}}
\int_{\mathbb{C}^{2}} e^{-q_1\overline{q}_1-q_2\overline{q}_2}\left(q_1\overline{q}_2+|c|^2\right)^{N-2}\left(q_2\overline{q}_1+|c|^2\right)^{N-2}\, \nonumber\\
&
\Big [ \left( \mathcal{K}_{0}- \frac{\mathcal{K}_{1}}{p_{1}p_{2}} \right) 
 |c|^{2} \left(q_{1}\bar q_{2} |q_{1}|^{2} +q_{2}\bar q_{1} |q_{2}|^{2}+
|c|^{2} (|q_{1}|^{2}+|q_{2}|^{2})\right) \nonumber\\
&
+ \frac{\mathcal{K}_{2}}{p_{1}p_{2}}  \, \left(  |q_{1}|^{2}  |q_{2}|^{2}  + |c|^{2} (1+|c|^{2}) (q_{1} \bar q_{2} +q_{2} \bar q_{1} )
\right) + \mathcal{K}_{1}\ |c|^{4} \left( \frac{q_{1}\bar q_{2}+ |c|^{2}}{p_{1}^{2}} +
 \frac{q_{2}\bar q_{1}+ |c|^{2} }{p_{2}^{2}}\right) \nonumber \\
&+\mathcal{K}_{0}\ |c|^{2}\left(  \left(\frac{p_{1}}{p_{2}}+\frac{p_{2}}{p_{1}} \right) |q_{1}|^{2}|q_{2}|^{2} + |c|^{2}
\left(\frac{p_{2}}{p_{1}}  q_{2}\bar q_{1}+\frac{p_{1}}{p_{2}} q_{1}\bar q_{2}\right)
\right)\Big]
\tfrac{dq_1d\overline{q}_1dq_2d\overline{q}_2}{\pi^2},\label{eq:coroll-proof}
\end{align}
where
\begin{align*}
 \mathcal{K}_{0}&=
 \int_{\mathbb{C}^{2N}} \delta(p_1-{\bf v}_1^*{\bf v}_1)\delta(p_2-{\bf v}_2^*{\bf v}_2)\ d{\bf v}_1^*d{\bf v}_1d{\bf v}_2^*d{\bf v}_2\\
 \mathcal{K}_{1}&=
 \int_{\mathbb{C}^{2N}} ( {\bf v}_1^*{\bf v}_2) ( {\bf v}_2^*{\bf v}_1)\,  \delta(p_1-{\bf v}_1^*{\bf v}_1)\delta(p_2-{\bf v}_2^*{\bf v}_2)\ d{\bf v}_1^*d{\bf v}_1d{\bf v}_2^*d{\bf v}_2\\
\mathcal{K}_{2}&=
 \int_{\mathbb{C}^{2N}} (c^{2} ( {\bf v}_1^*{\bf v}_2)^{2}+\bar c^{2} ( {\bf v}_2^*{\bf v}_1 )^{2})  \delta(p_1-{\bf v}_1^*{\bf v}_1)\delta(p_2-{\bf v}_2^*{\bf v}_2)\ d{\bf v}_1^*d{\bf v}_1d{\bf v}_2^*d{\bf v}_2.
\end{align*}
A systematic method to evaluate integrals of such type, in a more general setting,
was developed in \cite{YF2002}. Here the formulas are simple enough to do a direct computation.
For any continuous function
$f =f(u,\bar u),$ with  $u\in \mathbb{C},$ it holds
\begin{align*}
&\int _{\mathbb{C}^{2N}} f ({\bf v}_1^*{\bf v}_2,  {\bf v}_2^*{\bf v}_1)
\delta(p_1-{\bf v}_1^*{\bf v}_1)\delta(p_2-{\bf v}_2^*{\bf v}_2)\ d{\bf v}_1^*d{\bf v}_1d{\bf v}_2^*d{\bf v}_2\\
&=\quad \lim_{\varepsilon \to 0+}
\int _{\mathbb{C}^{2N}} f ({\bf v}_1^*{\bf v}_2,  {\bf v}_2^*{\bf v}_1)e^{-\varepsilon ({\bf v}_1^*{\bf v}_1+{\bf v}_2^*{\bf v}_2)}
\delta(p_1-{\bf v}_1^*{\bf v}_1)\delta(p_2-{\bf v}_2^*{\bf v}_2)\ d{\bf v}_1^*d{\bf v}_1d{\bf v}_2^*d{\bf v}_2
\\
&= \quad  \lim_{\varepsilon \to 0+}\int_{\mathbb{R}^{2}}  e^{itp_{1}+isp_{2}}
\int _{\mathbb{C}^{2N}} f ( {\bf v}_1^*{\bf v}_2,  {\bf v}_2^*{\bf v}_1) e^{- (\varepsilon +it) {\bf v}_1^*{\bf v}_1 }
e^{- (\varepsilon +is) {\bf v}_2^*{\bf v}_2 }d{\bf v}_1^*d{\bf v}_1d{\bf v}_2^*d{\bf v}_2 \frac{dt ds}{(2\pi)^{2}}.
\end{align*}
Inserting this formula in the integrals above we obtain $\mathcal{K}_{2}=0$ and
\begin{align*}
\mathcal{K}_{0}&=  \pi^{2N} \lim_{\varepsilon \to 0+}\int_{\mathbb{R}^{2}} \frac{e^{itp_{1}}}{(\varepsilon +it)^{N}}  \frac{e^{isp_{2}}}{(\varepsilon +is)^{N}} \frac{dt ds}{(2\pi)^{2}}=
\frac{\pi^{2N}}{(N-1)!^{2}} ( p_{1}p_{2})^{N-1}\\
\mathcal{K}_{1}&= N \pi^{2N} \lim_{\varepsilon \to 0+}\int_{\mathbb{R}^{2}} \frac{e^{itp_{1}}}{(\varepsilon +it)^{N+1}}  \frac{e^{isp_{2}}}{(\varepsilon +is)^{N+1}} \frac{dt ds}{(2\pi)^{2}}=
\frac{ \pi^{2N}}{N!(N-1)!} ( p_{1}p_{2})^{N}.
\end{align*}
The claim now follows by inserting these constants in \eqref{eq:coroll-proof} and reorganizing the resulting expression
in powers of $|c|.$
\end{proof} 

\section{Large-$N$ asymptotics}

\subsection{Proof of Theorem \ref{main_asy}}
\label{sec:3.1}
Our goal is to study the limit $N\to \infty$ of the ratio
\[
\pi_{N} (0,p_{1},p_{2})= \frac{\mathbb{E}[\Pi_{N} (0,p_{1},p_{2})]}{\mathbb{E}[\rho_{N} (0)]}=\frac{\overline{\Pi}_{N} (0,p_{1},p_{2})}{\int_{0}^{1}\overline{\Pi}_{N}(0,p_{1},1-p_{1})\, dp_{1}} \delta (p_{1}+p_{2}-1)
\]
(cf. Corollary \ref{Cor2.2}). To this end, we need to extract the leading asymptotic behaviour of the integral in \eqref{final_finite_size1},
assuming $|c|$ to be fixed and letting $N\to \infty$.
It is convenient to introduce the rescaled polar coordinates  $q_{1}=\sqrt{NR_{1}}e^{i\theta_1}, $
$q_{2}=\sqrt{NR_{2}}e^{i\theta_2},$ with $R_1,R_2\ge 0$ and $0\le \theta_1,\theta_2<2\pi$. Note that
$d\bar q_{i} dq_{i}=\frac{N}{2} dR_{i} d\theta_{i}.$
The integrand in \eqref{final_finite_size1}  depends on $R_{1},R_{2}$ and $\theta=\theta_{1}-\theta_{2}$ only, hence we obtain
\begin{align*}
\overline{\Pi}_{N} (0,p_{1},p_{2})&
= 
\frac{ N^{2N-2}}{(N-1)!(N-2)!} 
\frac{e^{-|c|^2\left(\frac{p_1}{p_2}+\frac{p_2}{p_1}\right)}}
{p_1p_2}\,\frac{1 }{\pi}\int_{(0,\infty )^{2}}  e^{-N (R_{1}+R_{2})} (R_{1}R_{2})^{N-2}\\
&\qquad \times \int_{0}^{2\pi }\left(1+\tfrac{|c|^2}{N\sqrt{R_1R_2}} e^{i\theta}\right)^{N-2}
\left(1+\tfrac{|c|^2}{N\sqrt{R_1R_2}}e^{-i\theta}\right)^{N-2}\\
&\qquad \qquad \times \mathcal{U}\Big ( NR_1,NR_2,N\sqrt{R_1R_{2}} e^{i\theta },N\sqrt{R_{1}R_2} e^{-i\theta} \Big)\ \frac{d\theta }{2\pi } dR_{1}dR_{2}.
\end{align*}
Extracting the $p_{1},p_{2}$ dependence explicitly we obtain
\begin{align}\label{eq:final_finite_size-scaled}
\overline{\Pi}_{N} (0,p_{1},p_{2})&=  \mathcal{N}_N \ \frac{1}{p_{1}p_{2}}e^{-|c|^2\left(\frac{p_1}{p_2}+\frac{p_2}{p_1}\right)} \left[ \mathcal{I}_{0,N}+\mathcal{I}_{1,N} \frac{p_{1}}{p_{2}}+ \mathcal{I}_{2,N} \frac{p_{2}}{p_1} \right]\\
\mathbb{E}[\rho_{N} (0)]&=\mathcal{N}_N\  \left[\mathcal{I}_{0,N} \mathcal{J}_{0}+
( \mathcal{I}_{1,N}+ \mathcal{I}_{2,N}) \mathcal{J}_{1} \right] ,\label{eq:final_finite_size-rho-scaled}
\end{align}
where, for $j=0,1,2,$
\begin{align}\label{eq:defINj}
 &\mathcal{I}_{j,N}=\mathcal{I}_{j,N} (|c|^{2}) :=\frac{N}{2\pi }\int_{(0,\infty )^{2}}  e^{-N (R_{1}-1-\ln R_{1})}e^{-N (R_{2}-1-\ln R_{2})} F_{j,N} (R_{1},R_{2})\, dR_{1} dR_{2},\\
&F_{j,N} (R_{1},R_{2}):=\tfrac{1}{(R_{1}R_{2})^{2}}
\hspace{-0,1cm} \int_{0}^{2\pi }\hspace{-0,2cm}
\left(1+\tfrac{|c|^2}{N\sqrt{R_1R_2}} e^{i\theta}\right)^{N-2}\left(1+\tfrac{|c|^2}{N\sqrt{R_1R_2}}e^{-i\theta}\right)^{N-2}\hspace{-0,2cm} \mathcal{U}_{j,N} (R_{1},R_{2},\theta )\  \frac{d\theta }{2\pi },
\nonumber
\end{align}
and the functions $\mathcal{U}_{j,N}(R_{1},R_{2},\theta ),$ $j=0,1,2,$ are defined by
\begin{align}\label{eq:Ujdef}
& \mathcal{U}_{0,N}:=  |c|^{2} \sqrt{R_{1}R_{2}}   (R_{1}e^{-i\theta }+R_{2}e^{i\theta }) +\frac{1}{N}|c|^{4} (R_{1}+R_{2})\\
&  \mathcal{U}_{1,N}=  |c|^{2} \frac{N}{N-1} R_{1}R_{2}+ |c|^{4} \frac{1}{N-1}  \sqrt{R_{1}R_{2}}  (e^{-i\theta }+e^{i\theta })-\frac{1}{N}|c|^{4}  \sqrt{R_{1}R_{2}}e^{i\theta } +
\frac{|c|^{6}}{N^{2} (N-1)}  \nonumber \\
&  \mathcal{U}_{2,N}=  |c|^{2} \frac{N}{N-1} R_{1}R_{2}+ |c|^{4} \frac{1}{N-1}  \sqrt{R_{1}R_{2}}  (e^{-i\theta }+e^{i\theta }) -\frac{1}{N} |c|^{4}  \sqrt{R_{1}R_{2}}e^{-i\theta }+
\frac{|c|^{6}}{N^{2} (N-1)}. \nonumber
\end{align}
The $N-$independent factors $\mathcal{J}_{i}$ come from the integration over $p_{1}:$ 
\begin{align*}
\mathcal{J}_{0}&=\mathcal{J}_{0}(|c|^{2}) := \int_{(0,1)^{2}}  \frac{1}{p_{1}p_{2}}e^{-|c|^2\left(\frac{p_1}{p_2}+\frac{p_2}{p_1}\right)}\delta (p_{1}+p_{2}-1)dp_{1}dp_{2},\\
\mathcal{J}_{1}&=\mathcal{J}_{1} (|c|^{2})= \frac{1}{2} \int_{(0,1)^{2}}  \frac{1}{p_{1}p_{2}}e^{-|c|^2\left(\frac{p_1}{p_2}+\frac{p_2}{p_1}\right)}\left(\frac{p_{1}}{p_{2}}+\frac{p_{2}}{p_{1}} \right)\delta (p_{1}+p_{2}-1)dp_{1}dp_{2}= -\frac{1}{2} \mathcal{J}_{0}' (|c|^{2}).
\end{align*}
Finally the $N-$dependent prefactor is 
\[
\mathcal{N}_N:=\frac{2 N^{2N-2}e^{-2N}}{(N-1)!(N-2)!},
\]
Note that, by Stirling's formula
\[
\lim_{N\to \infty} \mathcal{N}_N= \frac{1}{\pi }.
\]
\paragraph{\textbf{Claim.}}
 The  expressions in \eqref{eq:final_finite_size-scaled} and \eqref{eq:final_finite_size-rho-scaled} have the following
 asymptotic form, which directly implies the proof of the Theorem:
\begin{align}\label{eq:asymptnum}
&\lim_{N\to \infty}\left[  \mathcal{I}_{0,N}+\mathcal{I}_{1,N} \frac{p_{1}}{p_{2}}+ \mathcal{I}_{2,N} \frac{p_{2}}{p_1}\right]=
2|c|^{2}\left[  I_{1} (2|c|^2)+ \frac{1}{2}\left(  \frac{p_{1}}{p_{2}} +\frac{p_{2}}{p_1} \right)I_{0} (2|c|^2)  \right]\\
&\lim_{N\to \infty} \left[  \mathcal{I}_{0,N} \mathcal{J}_{0}+
( \mathcal{I}_{1,N}+ \mathcal{I}_{2,N}) \mathcal{J}_{1}  \right]= 2,\label{eq:asymptden}
\end{align}
where for $n\in \mathbb{N}_{0},$ $I_{n}$ is the modified Bessel function of the first kind (see \eqref{Bessel_I}).

To prove the claim we derive the 
  asymptotic form of the integrals $\mathcal{I}_{j,N}$  by Laplace method. 
Heuristically, the function $f (R_{j}):= R_{j}-1-\ln R_{j}$ has a unique global minimum in $R_{j}=1$ with $f'' (1)=1,$ hence
\begin{equation}\label{eq:asympt-IN} 
\lim_{N\to \infty} \mathcal{I}_{j,N}= \lim_{N\to \infty} F_{j,N} (1,1)  \int_{\R^{2}} e^{-\frac{1}{2} (x^{2}+y^{2})} \frac{dxdy}{2\pi } =
 \lim_{N\to \infty} F_{j,N} (1,1).
\end{equation}
The rigorous proof of this formula is given in the Appendix.
Using $ \lim_{N\to \infty} \mathcal{U}_{0,N} (1,1,\theta )= 2|c|^{2} \cos \theta, $ and $\lim_{N\to \infty} \mathcal{U}_{j,N} (1,1,\theta )=|c|^{2}$ for $j=1,2,$ and
\begin{equation}\label{exponentiating}
\lim_{N\to \infty}\left(1+\tfrac{|c|^2}{N\sqrt{R_1R_2}}e^{-i\theta}\right)^{N-2}\left(1+\tfrac{|c|^2}{N\sqrt{R_1R_2}}e^{i\theta}\right)^{N-2}=
e^{2|c|^2\cos{\theta}},
\end{equation}  
uniformly in $\theta\in [0,2\pi],$ we finally obtain
\begin{align*}
&\lim_{N\to \infty} \mathcal{I}_{0,N}=  |c|^{2} \int_{0}^{2\pi } e^{2|c|^2\cos{\theta}}2\cos \theta\
\tfrac{d\theta}{2\pi }=2|c|^{2}I_{1} (2|c|^2)
\\
&\lim_{N\to \infty} \mathcal{I}_{1,N}=\lim_{N\to \infty} \mathcal{I}_{2,N}=|c|^{2} \int_{0}^{2\pi } e^{2|c|^2\cos{\theta}} \tfrac{d\theta}{2\pi }=
|c|^{2}I_{0} (2|c|^2),
\end{align*}
which implies the asymptotic formula for the numerator \eqref{eq:asymptnum}.
To prove   \eqref{eq:asymptden} we notice that,
\begin{align*}
&\mathcal{J}_{0}=\int_0^{\infty}\int_0^{\infty}\frac{e^{-|c|^2\left(\frac{p_1}{p_2}+\frac{p_2}{p_1}\right)}}{p_1p_2}\,\delta(p_1+p_2-1)\,dp_1dp_2 \\
&\quad =
2\int_{-\infty}^{\infty}e^{-2|c|^2\cosh{2\alpha}}\int_0^{\infty}\delta(2p\cosh{\alpha}-1)\frac{dp}{p}\,d\alpha=
4\int_{0}^{\infty}e^{-2|c|^2\cosh{2\alpha}}d\alpha=2 K_0(2|c|^2),
\end{align*}
where in the second line we performed the coordinate change  $p_1=pe^{\alpha}, \, p_2=pe^{-\alpha}$
where $p\ge 0$ and $\alpha\in \mathbb{R}$  
and, for $n\in \mathbb{N}_0,$ $K_{n}$ the modified Bessel function of second kind (also known as Macdonald function)
\begin{equation}\label{eq:besselsecond}
K_{n} (x)= \int_{0}^{\infty}e^{-x\cosh{\alpha}}\cosh{n\alpha}\  d\alpha.
\end{equation}
Inserting this in the formulas above we obtain $\mathcal{J}_{1} (|c|^{2})= -2K_{0}' (2|c|^{2})= 2 K_{1} (2|c|^{2})$ and hence
\begin{align*}
&\lim_{N\to \infty} \left[ \mathcal{I}_{0,N} \mathcal{J}_{0}+
( \mathcal{I}_{1,N}+ \mathcal{I}_{2,N}) \mathcal{J}_{1} \right] =
2|c|^{2}\left[  I_{1} (2|c|^{2}) \mathcal{J}_{0} (|c|^{2}) + I_{0} (2|c|^{2}) \mathcal{J}_{1} (|c|^{2}) \right]\\
&= 4|c|^{2}\left[  I_{1} (2|c|^{2}) K_{0} (2|c|^{2}) + I_{0} (2|c|^{2}) K_{1} (2|c|^{2})\right]=2,
\end{align*}
where in the last line we used  the well-known Wronskian identity satisfied by modified Bessel functions (see e.g. \cite{WronskBessel})
\[
 I_{1} (2|c|^{2}) K_{0} (2|c|^{2}) + I_{0} (2|c|^{2}) K_{1} (2|c|^{2}) = \frac{1}{2|c|^{2}}.
\]

\subsection{Asymptotics assuming the extensive coupling: $c=\sqrt{N}\tilde{c}$}   
\label{subs:asext}

Now for completeness we briefly consider the regime when $|c|^2=O(N)$, where the replacement Eq.(\ref{exponentiating}) ceases to be valid.
In such a regime the corresponding factors non-trivially modify the points of maximum of the integrand in variables $R_1$ and $R_2$.
To take this change appropriately 
into account we rewrite the Eq.(\ref{eq:final_finite_size-rho-scaled}) as
\begin{align}
\mathbb{E}[\rho_{N} (0)]&=\mathcal{N}_N\  \left[\tilde{\mathcal{I}}_{0,N} \tilde{\mathcal{J}}_{0,N}+
( \tilde{ \mathcal{I}}_{1,N}+  \tilde{\mathcal{I}}_{2,N})  \tilde{\mathcal{J}}_{1,N} \right] 
\end{align}
where, for $j=0,1$
\[
 \tilde{\mathcal{J}}_{j,N}:=\left(\tfrac{N}{\pi} \right)^{\frac{1}{2}} e^{2N|\tilde{c}|^{2}}\mathcal{J}_{j}(N|\tilde{c}|^{2})=
 \left(\tfrac{N}{\pi} \right)^{\frac{1}{2}} e^{2N|\tilde{c}|^{2}} 2K_{j} (2N|\tilde{c}|^{2}),
\]
and for $j=0,1,2,$
\begin{align}
 &\tilde{\mathcal{I}}_{j,N} :=\tfrac{1}{\sqrt{2}}\left(\tfrac{N}{2\pi} \right)^{\frac{3}{2}}\int_{(0,\infty )^{2}\times [0,2\pi ]}
 e^{-N\mathcal{L}(R_1,R_2,\theta)}
\tfrac{\mathcal{W}_{j,N}(R_1,R_2,\theta)}{\left(R_1R_2+2|\tilde{c}|^2\sqrt{R_1R_2}\cos{\theta}+|\tilde{c}|^4\right)^2} 
dR_{1}dR_{2}d\theta ,
\label{finite_size_rescaled_new}
\end{align}
where
\begin{equation}\label{action}
\mathcal{L}(R_1,R_2,\theta)=R_1+R_2-2 (1-|\tilde{c}|^2) -\ln{\left(R_1R_2+2|\tilde{c}|^2\sqrt{R_1R_2}\cos{\theta}+|\tilde{c}|^4\right)}
\end{equation} 
and the functions $\mathcal{W}_{j,N}(R_{1},R_{2},\theta ),$ $j=0,1,2,$ are defined by
\begin{align}
& \mathcal{W}_{0,N}:=  |\tilde{c}|^{2} \sqrt{R_{1}R_{2}}   (R_{1}e^{-i\theta }+R_{2}e^{i\theta }) +|\tilde{c}|^{4} (R_{1}+R_{2})\nonumber\\
&  \mathcal{W}_{1,N}=  |\tilde{c}|^{2} \frac{N}{N-1} R_{1}R_{2}+ |\tilde{c}|^{4} \frac{N}{N-1}  \sqrt{R_{1}R_{2}}  (e^{-i\theta }+e^{i\theta })-|\tilde{c}|^{4}  \sqrt{R_{1}R_{2}}e^{i\theta } +
\frac{|\tilde{c}|^{6}}{(N-1)}  \nonumber \\
&  \mathcal{W}_{2,N}= \overline{ \mathcal{W}_{1,N}}. \nonumber
\end{align}
Using the asymptotic formula $K_{n} (z)\sim \sqrt{ \frac{\pi }{2z} }e^{-z}$ for $z\gg 1$ we get
\[
\lim_{N\to \infty}  \tilde{\mathcal{J}}_{j,N}= \frac{1}{|\tilde{c}|}.
\]
To study the asymptotics of $  \tilde{\mathcal{I}}_{j,N}$ we use again Laplace method. We distinguish two cases.\vspace{0,2cm}

For $|\tilde{c}|<1$ the function  $\mathcal{L}(R_1,R_2,\theta)$ admits a unique global minimum in $R_{1}=R_{2}=R_{m}:=1-|\tilde{c}|^{2},$
and $\theta=0,$ with Hessian matrix at the minimum given by
\[
\mathcal{L}'' (R_{m},R_{m},0)=
\begin{pmatrix}
\frac{(1+R_{m})}{2R_{m}} & -\frac{(1-R_{m})}{2R_{m}}& 0\\
 -\frac{(1-R_{m})}{2R_{m}}& \frac{(1+R_{m})}{2R_{m}}& 0\\
0&0& 2|\tilde{c}|^{2}R_{m}\\
\end{pmatrix}
\]
whose determinant is $\det\mathcal{L}'' (R_{m},R_{m},0) = 2|\tilde{c}|^{2}.$
Hence
\[
\lim_{N\to \infty}  \tilde{\mathcal{I}}_{0,N}= \frac{1}{\sqrt{2}\sqrt{ 2|\tilde{c}|^{2} }} \lim_{N\to \infty} \mathcal{W}_{0,N} (R_{m},R_{m},0)=
\frac{ 2R_{m}|\tilde{c}|^{2} }{2|\tilde{c}|}= |\tilde{c}|R_{m},
\]
and for $j=1,2$
\[
\lim_{N\to \infty}  \tilde{\mathcal{I}}_{j,N}= \frac{1}{\sqrt{2}\sqrt{ 2|\tilde{c}|^{2} }} \lim_{N\to \infty} \mathcal{W}_{j,N} (R_{m},R_{m},0)=\frac{|\tilde{c}|R_{m}}{2}.
\]
 The proof uses the same type of arguments as in the Appendix, but is simpler
since no separate analysis of the region near the boundaries $R_{1}=0, R_{2}=0$ is necessary.
Putting all this together we obtain, for $0<|\tilde c|<1$,
\[
\lim_{N\to \infty}\mathbb{E}[\rho_N(0)]=\frac{1}{\pi}\left(R_{m}+2\frac{R_{m}}{2}\right)= \frac{2}{\pi} (1-|\tilde{c}|^2).
\]
The endpoint $\tilde c=0$ follows directly from the uncoupled Ginibre case and gives the same formula. Note that the density vanishes as $|\tilde{c}|\to 1$ from below. The critical case $|\tilde c|=1$, where the interior saddle collides with the boundary, is not covered by the present Laplace analysis.

For $|\tilde{c}|> 1$ the function  $\mathcal{L}(R_1,R_2,\theta)$ is minimized at the boundary  $R_1=R_2=0.$ We compute
\[
\mathcal{L}_{m}:=\mathcal{L}(0,0,\theta)= 2 ( |\tilde{c}|^{2}-1 -\ln  |\tilde{c}|^{2})>0.
\]
Hence, for all $N>2$ and  $|\tilde{c}|> 1,$
\begin{align*}
&\left| \tilde{\mathcal{I}}_{j,N} \right|\leq N^{\frac{3}{2}} e^{- (N-2)\mathcal{L}_{m} } \int_{(0,\infty )^{2}\times [0,2\pi ]}
 e^{- (N-2)[\mathcal{L}(R_1,R_2,\theta)-\mathcal{L}_{m}]} |Pol (\sqrt{R_{1}},\sqrt{R_{2}}) |
dR_{1}dR_{2}d\theta \\
&\qquad \leq  N^{\frac{3}{2}}e^{- (N-2)\mathcal{L}_{m} } \int_{(0,\infty )^{2}\times [0,2\pi ]}
 e^{- [\mathcal{L}(R_1,R_2,\theta)-\mathcal{L}_{m}]} |Pol (\sqrt{R_{1}},\sqrt{R_{2}}) |dR_{1}dR_{2}d\theta\\
&\qquad  =N^{\frac{3}{2}}  e^{- (N-2)\mathcal{L}_{m} }C_{m},
\end{align*}
for some constant $C_{m}>0$ depending only on $\tilde{c},$
which finally yields 
\[
\lim_{N\to \infty}\mathbb{E}[\rho_N(0)]=0\qquad \forall |\tilde{c}|>1.
\]

This result has a simple physical interpretation.  For $|\tilde{c}|>1$ the vanishing density at $z=0$ is compatible with the expected splitting of the limiting spectral support into two disconnected components, with the origin lying outside the bulk.  In that expected global picture the two components meet at the origin as $|\tilde c|\downarrow1$ and merge for smaller coupling.  The calculation above establishes the behavior of the density at the origin; the global topology of the support is not derived here.

\appendix

\section{Asymptotic formulas}
Recall the definition of $\mathcal{I}_j,N$ in \eqref{eq:defINj}.
\begin{align*}
 &\mathcal{I}_{j,N}=\mathcal{I}_{j,N} (|c|^{2}) :=\tfrac{N}{2\pi }\int_{(0,\infty )^{2}}  e^{-N (R_{1}-1-\ln R_{1})}e^{-N (R_{2}-1-\ln R_{2})} F_{j,N} (R_{1},R_{2})\, dR_{1} dR_{2},\\
&F_{j,N} (R_{1},R_{2}):=\tfrac{1}{(R_{1}R_{2})^{2}}
\hspace{-0,1cm} \int_{0}^{2\pi }\hspace{-0,2cm}  \left(1+\tfrac{|c|^2}{N\sqrt{R_1R_2}} e^{i\theta}\right)^{N-2}\left(1+\tfrac{|c|^2}{N\sqrt{R_1R_2}}e^{-i\theta}\right)^{N-2}\hspace{-0,2cm} \mathcal{U}_{j,N} (R_{1},R_{2},\theta )\  \frac{d\theta }{2\pi },
\end{align*}
where $\mathcal{U}_{j,N} (R_{1},R_{2},\theta )$  are defined \eqref{eq:Ujdef}.
The goal of this section is to prove the formula
\begin{equation} 
\lim_{N\to \infty} \mathcal{I}_{j,N}= F_{j,\infty}:=
 \lim_{N\to \infty} F_{j,N} (1,1), \qquad j=0,1,2.
\end{equation}
We expect the integral $\mathcal{I}_{j,N}$ to be concentrated near the global minimum of the function
$R_{1}-1-\ln R_{1}+R_{2}-1-\ln R_{2}$ at $(R_{1},R_{2})= (1,1).$ Let $\delta:= N^{-\alpha }$  with
$\frac{1}{3}<\alpha <\frac{1}{2}.$ We decompose the integration domain as 
\[
(0,\infty )\times (0,\infty ) = \bigcup_{l=0}^{3} D_{l} 
\]
where $D_{0}$ is the region near the minimum $D_{0}:=\{(R_{1},R_{2})\in (1-\delta, 1+\delta  )\times  (1-\delta, 1+\delta  ) \} $
and the region far from the saddle
$D_{0}^{c}= D_{1}\cup D_{2}\cup D_3$ is decomposed as
\begin{align*}
D_{1}&:= \{(R_{1},R_{2})\in  (0,1-\delta]\times (0,1-\delta] \} \\
D_{2}&:=  \{(R_{1},R_{2})\in  (0,1-\delta]\times[1-\delta,\infty ) \} \cup
\{(R_{1},R_{2})\in [1-\delta,\infty ) \times  (0,1-\delta]\}\\
D_{3} &:=  \{(R_{1},R_{2})\in  [1-\delta,\infty )\times[1+\delta,\infty ) \}\cup 
 \{(R_{1},R_{2})\in [1+\delta,\infty )\times  [1-\delta,\infty ) \}.
\end{align*}
These regions distinguish the cases when one of the two variables is near to the boundary $R_{j}=0.$
The asymptotics \eqref{eq:asympt-IN}  follows directly from the following lemmas.

\begin{lemma}\label{Le:A1} There exists  $p\in \mathbb{N},$  and  constants $C_{0},N_{0}>0$ such that
\begin{align*}
\sup_{l=1,2,3}\sup_{j=0,1,2} \tfrac{N}{2\pi }&\int_{D_{l}}  e^{-N (R_{1}-1-\ln R_{1})}e^{-N (R_{2}-1-\ln R_{2})} \left| F_{j,N} (R_{1},R_{2})\right|\, dR_{1} dR_{2}  \ \leq \ C_{0}N^{p} e^{-\frac{\delta^{2}N}{4}}\\
&= C_{0}N^{p} e^{-\frac{N^{1-2\alpha }}{4}}\qquad \forall N\geq N_{0}.
\end{align*}
\end{lemma}

\begin{lemma}\label{Le:A2}  For $N\gg 1$ it holds
\begin{align*}
 \tfrac{N}{2\pi }&\int_{D_{0}}  e^{-N (R_{1}-1-\ln R_{1})}e^{-N (R_{2}-1-\ln R_{2})}  F_{j,N} (R_{1},R_{2})\, dR_{1} dR_{2} =
\bigl(F_{j,\infty}+O(N^{-1})\bigr) \Bigl(1+ O (N\delta^{3})+O (\delta )\Bigr)\\
& = \bigl(F_{j,\infty}+O(N^{-1})\bigr) \Bigl(1+ O (N^{- (3\alpha -1)})+O (N^{-\alpha} )\Bigr)
\end{align*}
for all $j=0,1,2.$ 
\end{lemma}

\begin{proof}[Proof of Lemma \ref{Le:A1}]
We can decompose $\mathcal{U}_{j,N}$ as $\mathcal{U}_{j,N}= \mathcal{U}_{j}+ \frac{1}{N} Q_{j},$ where
\begin{equation}\label{Udecomp1}
 \mathcal{U}_{0}=|c|^{2}\sqrt{R_{1}R_{2}}  \left(R_{1}e^{i\theta}  +R_{2} e^{-i\theta}  \right),\qquad
  \mathcal{U}_{1}=\mathcal{U}_{2}= |c|^{2} R_{1}R_{2}
  \end{equation}
and, for some constant $C'>0,$
\begin{equation}\label{Udecomp2}
\sup_{\theta \in [0,2\pi ]}\sup_{j=0,1,2}|Q_{j}|\leq  C' \left( (\sqrt{R_{1}} +\sqrt{R_{2}})^{2}+\frac{1}{N} \right).
\end{equation}
Therefore, using $|(a+b)|^{N}\leq ( |a|+|b|)^{N},$ the function $F_{j,N}$ is bounded by
\begin{align*}
|F_{j,N} (R_{1},R_{2})|\leq  \frac{1}{(R_{1}R_{2})^{2}}\left(1+\frac{|c|^2}{N\sqrt{R_1R_2}}\right)^{2 (N-2) }
|Pol (R_{1}) Pol ( R_{2})|.
\end{align*}
We argue in the three regions $D_{1},D_{2},D_{3}$ separately.\vspace{0,2cm}

\textit{Region $D_{1}.$} Here both variables are near zero $R_{1},R_{2}\leq 1-\delta$ and in particular bounded.
Therefore it is enough to bound an integral of the form
\begin{align*}
&\int_{(0,1-\delta )^{2}} e^{-N(R_1-1)-N(R_2-1) } (R_{1}R_{2})^{N-2} \left( 1+\frac{|c|^2}{N\sqrt{R_1R_2}} \right)^{2 (N-2)}dR_{1}dR_{2}\\
&\qquad \int_{(0,1-\delta )^{2}} e^{-N(R_1-1)-N(R_2-1) }  \left(\sqrt{R_1R_2} +\frac{|c|^2}{N} \right)^{2 (N-2)}dR_{1}dR_{2}\\
&\qquad \qquad \leq \frac{N^{4}}{|c|^{8}}\int_{(0,1-\delta )^{2}}\, e^{-NA (\sqrt{R_1},\sqrt{R_2})}dR_{1}dR_{2}
\end{align*}
where, for $x,y>0,$ 
\[
A (x,y):= x^{2}-1+y^{2}-1-2 \ln\left[  xy+\frac{|c|^{2}}{N} \right].
\]
This function has a unique global minimum at the point $x=y=\sqrt{1-\frac{|c|^{2}}{N}}.$
Since $\delta \gg \frac{1}{N},$ this point is outside the integration domain $D_{1},$
and hence in $D_{1}$ the minimum of $A$ is achieved at the boundary.
Direct inspection shows the minimum is achieved at the point $x=y=\sqrt{1-\delta}.$
We compute, using $N^{-1}=o (\delta^{2}),$ which holds since $\alpha <\frac{1}{2},$ 
\begin{align*}
NA (\sqrt{1-\delta},\sqrt{1-\delta} )&=N\left[ -2\delta -2 \ln\left ( 1-\delta +\frac{|c|^{2}}{N} \right)  \right] =
N\left[\delta^{2}+ o (\delta^{2}) \right]\geq \frac{ N\delta^{2}}{4}
\end{align*}
for $N$ large enough.
The result follows since $D_{1}$ is compact.\vspace{0,2cm}

\textit{Region $D_{2}.$} Here one variable (say $R_{1}$) is near zero $R_{1}\leq 1-\delta $ and one variable (say $R_{2}$)
is far from zero $R_{2}\geq 1-\delta.$

Since $R_{2}\geq 1-\delta $ and $R_{1}$ is bounded it is sufficient to consider an integral of the form
\[
\int_{(0,1-\delta )}   e^{-N(R_1-1)}   R_{1}^{N-2} \left( 1+\frac{|c|^2}{N\sqrt{R_1 (1-\delta )}} \right)^{2 (N-2)}
\int_{(1-\delta,\infty) } e^{-N(R_2-1-\ln R_{2})} |Pol (R_{2})| dR_{2} dR_{1}.
\]
We argue, since $R_2-1-\ln R_{2}\geq 0,$
\[
\int_{1-\delta}^{\infty } e^{-N(R_2-1-\ln R_{2})} |Pol (R_{2})| dR_{2}\leq \int_{\frac{1}{2}}^{\infty } e^{-(R_2-1-\ln R_{2})} |Pol (R_{2})| dR_{2}=C',
\]
for some $C'>0$ independent from $N.$
Hence, we are left with estimating
\[
\int_{(0,1-\delta )}   e^{-N(R_1-1)}   R_{1}^{N-2} \left( 1+\frac{|c|^2}{N\sqrt{R_1 (1-\delta )}} \right)^{2 (N-2)}
\leq \frac{N^{4}}{|c|^{8}} \int_{(0,1-\delta )}   e^{-N A (\sqrt{R_{1}})} dR_{1}.
\]
where we defined for $x>0$
\[
A (x):= x^{2}-1-2 \ln\left[  x+\frac{|c|^{2}}{N \sqrt{1-\delta }} \right].
\]
This function has the global minimum at the point $x =1-\frac{|c|^{2}}{2N \sqrt{1-\delta }}+o (N^{-1}) >\sqrt{1-\delta},$ since
$N^{-1}=o (\delta).$ Hence the minimum is achieved at the boundary $x=\sqrt{1-\delta}:$
\[
A (x)\geq  A (\sqrt{1-\delta})= -\delta -2   \ln\left[ (\sqrt{1-\delta})+\frac{|c|^{2}}{N \sqrt{1-\delta }} \right]
=\frac{\delta^{2}}{2} +o (\delta^{2})\geq \frac{\delta^{2}}{4}.
\]
for $N$ large enough. Hence $NA (x)\geq  NA (\sqrt{1-\delta})\geq  \frac{1}{4}N\delta^{2} .$
The result follows since $[0,1-\delta ]$ is compact.\vspace{0,2cm}

\textit{Region $D_{3}.$} Here both variables are far from zero ($R_{1},R_{2}\geq 1-\delta $) and at least one, say $R_{1},$
is far from the saddle $R_{1}\geq 1+\delta$.
Since $R_{1},R_{2}>1-\delta $ we have
\[
|F_{j,N} (R_{1},R_{2}) |\leq \tfrac{1}{(1-\delta)^{4}}\left(1+\tfrac{|c|^{2}}{N (1-\delta)^{2}} \right)^{(2N-4)}\hspace{-0,2cm} |Pol (R_{1})Pol (R_{2})  |=C' e^{2|c|^{2}+o (1)}|Pol (R_{1})Pol (R_{2})  |
\]
so we have to estimate an integral of the form 
\[
     \int_{1+\delta }^{\infty} e^{-N(R_1-1-\ln R_{1})} Pol (R_{1})  dR_1\,  \int_{1-\delta }^{\infty  } e^{-N(R_2-1-\ln R_{2})} Pol (R_{2})dR_2.
\]
Since $(R_2-1-\ln R_{2}\geq 0$ we can bound 
\[
 \int_{1-\delta  }^{\infty} e^{-N(R_2-1-\ln R_{2})} |Pol (R_{2})|  dR_2\leq
 \int_{\frac{1}{2} }^{\infty} e^{-(R_2-1-\ln R_{2})}  |Pol (R_{2})|  dR_2 =C',
\]
for some constant $C'>0$ independent of $N.$
Finally, to estimate the integral over $R_{1}$ we notice that
\[
\min_{R_{1}\geq 1+\delta }\left[R_{1}-1-\ln R_{1} \right]= (1+\delta )-1-\ln  (1+\delta )=\frac{\delta^{2}}{2}+o (\delta^{2}).
\]
Hence
\begin{align*}
& \int_{1+\delta }^{\infty  }dR_1\, e^{-N(R_1-1-\ln R_{1})}| Pol (R_{1})|= \int_{1+\delta }^{\infty  }dR_1\, e^{-\frac{N}{2}(R_1-1-\ln R_{1})} e^{-\frac{N}{2}(R_1-1-\ln R_{1})}| Pol (R_{1})|\\
&
\leq  e^{-\frac{N}{2}(\delta -\ln (1+\delta))} \int_{1+\delta }^{\infty  }dR_1\,  e^{-\frac{N}{2}(R_1-1-\ln R_{1})}| Pol (R_{1})|\\
&\leq
  e^{-\frac{N}{2}(\delta -\ln (1+\delta))} \int_{1/2 }^{\infty  }dR_1\,  e^{-(R_1-1-\ln R_{1})}| Pol (R_{1})|= C' \,
  e^{-\frac{N\delta^{2}}{4}  +o (1)}.
\end{align*}
This completes the proof of the lemma 
\end{proof}
\begin{proof}[Proof of Lemma \ref{Le:A2}]
 Since both variables are in the interval $[1-\delta,1+\delta ]$ and 
 $\delta\ll 1,$ we can expand all functions around $R_{1},R_{2}=1:$
\[
\frac{1}{(R_{1}R_{2})^{2}}\left(1+\tfrac{|c|^2}{N\sqrt{R_1R_2}} e^{i\theta}\right)^{N-2}\left(1+\tfrac{|c|^2}{N\sqrt{R_1R_2}}e^{-i\theta}\right)^{N-2}=
 e^{2|c|^2\cos\theta} e^{O (N^{-1}+\delta )}
\]
\[
N(R_j-1-\ln R_{j})= \frac{N}{2} (R_{j}-1)^{2}+O (N\delta^{3}). 
\]
Note that $N\delta^{3}\to 0$ as $N\to \infty$ since $\alpha >\frac{1}{3}.$
Finally, using \eqref{Udecomp1} and \eqref{Udecomp2}, we argue $\mathcal{U}_{j,N}=\mathcal{U}_{j}+O (N^{-1}) $ and
\[
 \mathcal{U}_{0}= |c|^{2} 2 \cos \theta +O (\delta ),\qquad  \mathcal{U}_{1}=\mathcal{U}_{2}= |c|^{2}+O (\delta ).
\]
Replacing $R_{j}-1$ with $x_{j}$ we get
\begin{align*}
&\tfrac{N}{2\pi }\int_{(1-\delta ,1+\delta  )^{2}}  e^{-N (R_{1}-1-\ln R_{1})}e^{-N (R_{2}-1-\ln R_{2})} F_{0,N} (R_{1},R_{2})\, dR_{1} dR_{2}\\
&= e^{O (\delta^{3}N)+O (\delta )+O (N^{-1})}
\tfrac{N}{2\pi }\int_{(-\delta ,+\delta  )^{2}}  e^{- \frac{N}{2} ( x_{1}^{2}+x_{2}^{2})}  dx_{1} dx_{2}
 \int_{0}^{2\pi }  e^{2|c|^2\cos\theta} \left[  |c|^{2} 2 \cos \theta +O (\delta )  \right]  \tfrac{d\theta }{2\pi } \\
&= e^{O (\delta^{3} N)+ O (\delta )+O (N^{-1})}  \left[ 2|c|^{2}I_{1} (2|c|^2)+O (\delta )  \right]. 
\end{align*}
Similarly, for $j=1,2$
\begin{align*}
&\tfrac{N}{2\pi }\int_{(1-\delta ,1+\delta  )^{2}}  e^{-N (R_{1}-1-\ln R_{1})}e^{-N (R_{2}-1-\ln R_{2})} F_{1,N} (R_{1},R_{2})\, dR_{1} dR_{2}\\
&\qquad  =  e^{O (\delta^{3}N)+O (\delta )+O (N^{-1})}  \left[ |c|^{2}I_{0} (2|c|^2)+O (\delta )  \right]. 
\end{align*}
This completes the proof of the lemma.
\end{proof}


\end{document}